\documentclass{article}
\usepackage{amsmath, amsfonts,amssymb,amsthm,comment}
\usepackage{cases}
\usepackage{mathtools}
\usepackage{hyperref}
\usepackage{subfig}
\usepackage{soul,color}
\soulregister\cite7
\soulregister\ref7
\soulregister\pageref7

\usepackage[hmargin=1in]{geometry}

\newtheorem{thm}{Theorem}[section]

\newtheorem{cor}{Corollary}[section]
\newtheorem{lem}{Lemma}[section]
\newtheorem{prop}{Proposition}[section]
\newtheorem{defi}{Definition}[section]
\newtheorem{Assumption}{Assumption}[section]
\newtheorem{remark}{Remark}[section]

\newcommand{\et}{\mathrm{e}}
\newcommand{\be}{\begin{equation}}
\newcommand{\ee}{\end{equation}}

\newcommand{\bq}{\begin{eqnarray}}
\newcommand{\eq}{\end{eqnarray}}
\newcommand{\ex}{\mathbb{E}}

\newcommand{\ind}{\mathbf{1}}
\newcommand{\nn}{\nonumber}

\newcommand{\sk}{\smallskip}
\newcommand{\ol}{\overline}
\newcommand{\mc}{\mathcal}
\newcommand{\p}{\partial}

\newcommand{\e}{\mathbf{e}}

\newcommand{\pr}{\mathbb{P}}

\newcommand{\diff}{\textup{d}}
\newcommand{\lev}{L\'evy }
\newcommand{\R}{\mathbb{R}}

\newcommand{\ans}{\textcolor[rgb]{0.00,0.00,0.72}}

\begin{document}

\title{\Large \textbf{On the optimality of threshold type strategies in single and recursive optimal stopping under \lev models}}

\author{Mingsi Long\footnote{Center of Data Science, New York University, New York, NY 10011,  Email: longmingsi01@gmail.com} \and
Hongzhong Zhang\footnote{
Department of IEOR, Columbia University, New York, NY 10027, USA. Email: {hz2244@columbia.edu}. }
}
\maketitle
\begin{abstract}
In the spirit of Surya \cite{Surya2007}, we develop an average problem approach  to prove the optimality of threshold type strategies for optimal stopping of  L\'evy models with a continuous additive functional (CAF) discounting. Under spectrally negative models, we specialize this  in terms of  conditions on the reward function and  random discounting, where we present two examples of local time and occupation time discounting. We then apply this approach to recursive optimal stopping problems, and present simpler and neater proofs for a number of important results on qualitative  properties of the optimal thresholds, which are only known under a few special cases \cite{Carmona2008,TimKazuHZ14, YamazakiAMO2014}.
\end{abstract}


\section{Introduction}\label{sec:intro}
Let $X_\cdot=(X_t)_{t\ge0}$ be a general L\'evy process, with c\`adl\`ag paths, living on a filtered probability space $(\Omega,\mathcal{F},\{\mathcal{F}_{t}\}_{t\geq0},\pr)$, where $\mc{F}$ is assumed as the augmented natural filtration of $X_\cdot$.
We study the optimality of threshold type strategies in a class of optimal stopping problems driven by $X_\cdot$. In particular, we consider an optimal single stopping problem with a continuous additive functional (CAF) random discounting, and a sequence of recursive optimal stopping problems that arise from pricing of swing options \cite{Carmona2008} and contraction options \cite{YamazakiAMO2014}. For all those problems, we show that a sufficient condition for the optimality of threshold type strategy can be formulated in terms of an auxiliary average problem about the running maximum\footnote{By considering the \lev process $-X_\cdot$, one can easily adjust the argument to incorporate the running minimum of $X_\cdot$.} of $X_\cdot$ at a doubly stochastic random time. When $X_\cdot$ is a spectrally negative L\'evy process, we show that this average problem can be explicitly solved  through the ``first-order condition'' equation. The optimality of threshold type strategies then follows from the monotonicity property of the solution to this average problem. Moreover, we also show that a main result in \cite{HLY14} on optimal stopping with log-concave rewards can be seen as a special case of our Theorem \ref{thmSNLP}.

This work generalizes similar ideas as in \cite{Alili2005, KyprSury05, mordecki2002, Surya2007}, where optimal stopping of a L\'evy process under a constant discounting rate were studied. They show that the optimality of up-crossing threshold type strategies follows from a special form of the reward function, namely, there exists a nondecreasing function $h(\cdot)$ such that the expectation of $h(\overline{X}_{\mathbf{e}_r})$ is equal to the reward function, where $\mathbf{e}_r$ is an independent exponential random variable with parameter $r\ge0$ equal to the discounting rate. In \cite{TimKazuHZ14}, the authors use  measure change techniques to generalize this approach to case of a negative discounting rate.  This approach is further applied in \cite{OmegaRZ15}, to evaluate a perpetual American call option with an occupation time type discounting. As seen in \cite[Section 4]{LTZIJTAF15},
knowing the optimality of threshold type strategies  and monotonicity of optimal thresholds effectively helps reduce a complicated stochastic optimization in optimal stopping problems to a parametric optimization, which can then be developed into an efficient numerical algorithm.

Our objective of this study is to present an application of the average problem approach under random discounting in a {\em general} setting. While the optimality of threshold type strategies often fails to hold in a random discounting setting (as shown in \cite{Dayanik2008b,OmegaRZ15}), it is still valuable to give conclusive answers to a wide class of problems using results in this work. A further and equally important consideration is to show that many important, yet difficult-to-prove qualitative results in optimal multiple stopping problems and recursive optimal stopping problems (e.g. optimality of threshold type strategies, monotonicity of optimal thresholds, etc) can be easily proved by following the average problem approach.  Although this approach  is known for years,
to our best knowledge, the above applications  are novel and will help the dissemination  of this powerful method in a wider class of problems in practice.

Optimal stopping problems with a random discounting find their applications  in many areas such as finance and applied probability (for instance, in problems driven by continuously time-changed Markov processes \cite{CPT12}, or problems with random maturity \cite{OmegaRZ15}).
Perpetual optimal stopping of time-homogeneous diffusions with a random discounting rate were studied in Dayanik \cite{Dayanik2008b}, by exploiting Dynkin's concave characterizations of the excessive functions. In particular, the author managed to directly ``construct'' the value function, without postulating (and verifying) any prior ansatz about the structure of the optimal stopping region. While this is a very promising approach within diffusion framework, it shows limitations when jumps present. For example, for perpetual American call options on exponential L\'evy models under an occupation time type discounting  \cite{OmegaRZ15},  it is shown that there can be two disjoint components for both the continuation and the stopping regions,  which appear alternately. Possible overshoots thus make it difficult to apply \cite{Dayanik2008b}'s approach directly.

Optimal stopping problems with multiple exercising and refraction times arise in many application in finance and operations research. For instance, Carmona and Touzi \cite{Carmona2008} formulated the valuation of a swing put option as optimal multiple stopping problem, with constant refraction periods, under the Black-Scholes model. In a related work, Zeghal and Mnif \cite{ZeghalSwing} priced a perpetual American swing put option under spectrally positive exponential L\'evy models. Later,
 Leung et al. \cite{TimKazuHZ14} considered a stock loan with multiple repayments as an optimal multiple stopping problem with negative discounting rate and general i.i.d. positive refraction periods.  Among them, \cite{Carmona2008} and \cite{ZeghalSwing} subsequently established the optimality of threshold type strategies  in exercising a perpetual American swing put option,\footnote{\cite[Proposition 3.1]{ZeghalSwing} proved the optimality of threshold type strategy using monotonicity and convexity of the value function, which, is not fully legitimate, because the reward functions in the multiple optimal stopping problems can also be {\em curved}, convex functions, leaving arguments based on a put payoff invalid.} and demonstrated these monotonicity of the optimal thresholds using  sub-gradients techniques under two special models. By using mathematical inductions and the supermartingale property of value functions,
\cite{TimKazuHZ14} proved  similar results for a multiple-exercising call option under a general L\'evy model with arbitrary negative jumps and Phase-type positive  jumps. Moreover,  optimal multiple stopping problems with a running cost were studied in Yamazaki \cite{YamazakiAMO2014} to address the optimal timing to withdraw from a project in stages. In particular, under spectrally negative \lev models, \cite{YamazakiAMO2014} explicitly calculated the value of down-crossing threshold type strategies in terms of the so-called scale functions,  which was then used in conjunction with smooth fit to show the optimality of threshold type strategies.

The remaining paper is structured as follows. 
In Section \ref{sec:single} we study the single optimal stopping problem with a random discounting by using the average problem approach. 
 In Section \ref{sec:SNLPs}, we specialize the L\'evy model  to that of spectrally negative processes, and give an explicit construction of the solution to the average problem.  To illustrate the idea, we present two examples in Section \ref{sec:Example}: a generalization of the local time discounting optimal stopping problem as studied in \cite{Dayanik2008b} to a spectrally negative $\beta$-stable process, and a case of the Novikov-Shiryaev problem (see \cite{KyprSury05,NovkShir07}) under an occupation time discounting.  In Section \ref{sec:mul}, we apply the average problem approach to recursive optimal stopping problem and give simpler proofs for some important results within this context. Specifically, under a general \lev model, we study the case with a deterministic discount rate and refraction times but without running cost in Section \ref{sec41}, and the case with a random discounting and  a running cost in Section \ref{sec42}. Omitted technical proofs can be found in Appendix \ref{sec:proof}.  Some useful facts about the scale functions of spectrally negative L\'evy processes are reviewed in Appendix \ref{sec:SNLP}.

Throughout the paper, we use $\pr_x$ and $\ex_x$ to denote the probability law and the corresponding expectation given $X_0=x$, and we will suppress the subscripts in $\pr_x$ and $\ex_x$ if $x=0$.

%
%
%
%

\section{Single optimal stopping problem with random discounting}\label{sec:single}

We consider the following optimal stopping problem:
\be\label{singleproblem}
V(x):=\sup_{\tau\in\mathcal{T}}\ex_{x}[\et^{-A_{\tau}}f(X_{\tau})\ind_{\{\tau<\infty\}}],
\ee
where $\mathcal{T}$ is the set of all $\mathcal{F}$-stopping times with values in $[0,\infty]$, $f(\cdot)$ is the reward function, which is lower semi-continuous, and satisfies Condition (M) in Definition \ref{conM0} below.\footnote{\label{fn:well}As seen in Theorem \ref{singlesolution} below, this will guarantee the problem \eqref{singleproblem} is well defined. }  Here $A_\cdot=(A_{t})_{t\geq0}$ is a continuous additive functional, or CAF,  of $X_\cdot$. Namely, $A_\cdot$ is an $\mathcal{F}$-adapted process that is almost surely non-negative, continuous, and satisfies 
$$
A_0=0,\quad A_{s+t}=A_s+A_t\circ\theta_s,\text{ }s,t\geq0\text{ a.s.}
$$
where $\theta_s$ is the usual Markov shifting operator, namely, $X_t\circ\theta_s=X_{t+s}$ for all $t,s\ge0$ (see \cite[page 133]{CAF66} for a more complete definition of CAF). From the definition it is clear that a CAF $A_\cdot$ is nondecreasing. Throughout, we make the following standing assumption.
\begin{Assumption}\label{ass:A}For all $x\in\R$, we have
\be
\textrm{either }\pr_x(A_\infty=\infty)=1\textrm{ or }\pr(\limsup_{t\to\infty}X_t<\infty)=1.\nn\ee
\end{Assumption}

Given an independent, unit mean exponential random variable $\mathbf{e}$, let us introduce the left inverse of $\mathbf{e}$ by $A_\cdot$:
\be
\zeta:=\inf\{t>0\,:\,A_t>\mathbf{e}\},\nn
\ee
where, as usual, we set $\inf\emptyset=\infty$. Because
\[\pr_x(\zeta=\infty)=\pr_x(A_t\le \mathbf{e}, \forall t>0)=\pr_x(A_\infty\le\mathbf{e})=\ex_x[\exp(-A_\infty)],\]
we have from Assumption \ref{ass:A} that, either (i) $\zeta<\infty$ almost surely; or (ii) $\limsup_{t\to\infty}X_t<\infty$ almost surely in case (i) fails to hold.

We Let $T_z^+$ be the first passage time of $X_\cdot$ over a given threshold $z$ from below, i.e.
\be
T_z^+:=\inf\{t>0: X_t>z\},\label{eq:def:Tz}
\ee
and denote the running maximum process by $\ol{X}_t:=\sup_{s\in[0,t]}X_s$.
Then the random variable $\ol{X}_\zeta$ is well defined and is finite a.s. and satisfies
\be\pr_x(\overline{X}_\zeta>z)=\pr_x(T_z^+<\zeta, T_z^+<\infty)=\pr_x(A_{T_z^+}<\mathbf{e}, T_z^+<\infty)=\ex_x[\exp(-A_{T_z^+})\ind_{\{T_z^+<\infty\}}],\,\,\forall z\ge x.\label{eq:p2e}
\ee
In fact, we have the following equivalence:
\begin{lem}\label{lem:equivalence}
Assumption \ref{ass:A} is equivalent to 
\be
\label{eq:condition1}
\lim_{z\to\infty}\ex_x[\exp(-A_{T_z^+})\ind_{\{T_z^+<\infty\}}]=0,\quad\forall x\in\R.
\ee
\end{lem}

\begin{defi}[Condition (M)]\label{conM0}
 We say the reward function $f(\cdot)$ satisfies Condition (M) with random time $\zeta$ if 
 there is a nondecreasing function $h$ such that $h(x)>0$ if and only if $x>x^\star$ for some constant $x^\star\in[-\infty,\infty)$, and it holds that
\be
\ex_x[|h(\ol{X}_{\zeta})|]<\infty, \quad f(x)=\ex_x[h(\ol{X}_{\zeta})],\quad\forall x\in\R,\label{eq:conM0}
\ee
where $\ol{X}_{\zeta}=\sup_{s\in[0,\zeta]}X_s$. Moreover, we will denote by $\Upsilon_{\zeta}$ the set of all reward functions satisfying Condition (M) with random time $\zeta$.
\end{defi}
\begin{remark}\label{rem:convex}
The set  $\Upsilon_{\zeta}$  is a convex cone. That is, if $f(\cdot),g(\cdot)\in\Upsilon_{\zeta}$, then
\begin{enumerate}
\item[(i)] $\alpha f(\cdot)\in\Upsilon_{\zeta}$,  for any $\alpha>0$;
\item[(ii)] $f(\cdot)+g(\cdot)\in\Upsilon_{\zeta}$.
\end{enumerate}
\end{remark}

\begin{remark}\label{remark:32}
In the case that the discount factor rate is a constant $r>0$, the random time $\zeta$ is the exponential random variable with mean $1/r$, which we denote as $\e_r$. Notice that, if $f(\cdot)\in\Upsilon_{\mathbf{e}_r}$, i.e., $f(x)=\ex_x[h(\ol{X}_{\mathbf{e}_r})]=\ex[h(x+\ol{X}_{\mathbf{e}_r})]$ for some nondecreasing function $h(\cdot)$, then we have $f(x)\le f(y)$ for any $x<y$.  So every element in
$\Upsilon_{\mathbf{e}_r}$ is nondecreasing. \end{remark}

Below we state our main result.

\begin{thm}\label{singlesolution}
Suppose the reward function $f(\cdot)$ is lower semi-continuous and belongs to $\Upsilon_{\zeta}$ (defined in Definition \ref{conM0}),
then we have \footnote{\label{fn:xinfty}If the reward function $f(\cdot)$ satisfies  Condition (M) but with an $x^\star=\infty$, then we have $f(x)\le 0$ for all $x\in\R$, so the value function $V(x)$ is trivially 0. In this case, Eq. \eqref{eq:singlesolution} still holds.}
\be
V(x)=\sup_{\tau\in\mathcal{T}}\ex_{x}[\et^{-A_{\tau}}f(X_{\tau})\ind_{\{\tau<\infty\}}]=\ex_x[h(\overline{X}_{\zeta})\ind_{\{\overline{X}_{\zeta}>x^{\star}\}}].\label{eq:singlesolution}
\ee
That is, the value function $V(\cdot)\in\Upsilon_\zeta$.
Besides, the optimal stopping time is the up-crossing strategy $\tau^{\star}=T_{x^{\star}}^+$. That is,
\be
\ex_x[h(\overline{X}_{\zeta})\ind_{\{\overline{X}_{\zeta}>x^{\star}\}}]=\ex_{x}[\exp(-A_{T_{x^{\star}}^{+}})f(X_{T_{x^{\star}}^{+}})\ind_{\{T_{x^{\star}}^{+}<\infty\}}].\nn
\ee
\end{thm}
\begin{proof}
Let us define function
\be
v(x):=\ex_x[h(\overline{X}_{\zeta})\ind_{\{\overline{X}_{\zeta}>x^{\star}\}}].\nn\ee
We first prove that $(\et^{-A_{t}}v(X_{t}))_{t\geq 0}$ is a $\pr_{x}$-supermartingale for any $x\in\mathbb{R}$.
As $A_{t}$ is additive, the random variable $\zeta$ has the key property of being memoryless and $\pr_{x}(\zeta>t|\mc{F}_s,\zeta>s)=\pr_{x}(A_s+A_{t-s}\circ\theta_s<\e|\mc{F}_s,A_s<\e)=\pr_{X_s}(A_{t-s}<\mathbf{e})=\ex_{X_s}[\et^{-A_{t-s}}]$ for $t>s$. On the event $\{t<\zeta\}$, the identity $\overline{X}_{\zeta}=\overline{X}_{t}\vee\sup_{s\in[t,\zeta]}X_s\ge \sup_{s\in[t,\zeta]}X_s$ holds almost surely, so does $h(\overline{X}_{\zeta})\ind_{\{\overline{X}_{\zeta}>x^\star\}}\ge h(\sup_{s\in[t,\zeta]}X_s)\ind_{\{\sup_{s\in[t,\zeta]}X_s>x^{\star}\}}$, thanks to the nondecreasing property of $h(\cdot)\ind_{\{\cdot>x^\star\}}$. Hence, for any $x\in\R$,
\begin{align}\nonumber
v(x)&=\ex_{x}[h(\overline{X}_{\zeta})\ind_{\{\overline{X}_{\zeta}>x^{\star}\}}]\\\nonumber
&\geq \ex_{x}\big[\ex_{x}[h(\overline{X}_{\zeta})\ind_{\{\overline{X}_{\zeta}>x^{\star}\}}\ind_{\{t<\zeta\}}|\mathcal{F}_{t}, t<\zeta]\big]\\\nn
&\geq \ex_{x}\big[\ind_{\{t<\zeta\}}\ex_{x}[h(\sup_{s\in[t,\zeta]}X_s)\ind_{\{\sup_{s\in[t,\zeta]}X_s>x^{\star}\}}|\mathcal{F}_{t}, t<\zeta]\big]\\\nonumber
&=\ex_{x}\big[\ind_{\{t<\zeta\}}\ex_{X_{t}}[h(M)\ind_{\{M>x^\star\}}]\big]\\\nonumber
&=\ex_{x}[\ind_{\{t<\zeta\}}v(X_{t})]=\ex_x[\exp(-A_t)v(X_t)],
\end{align}
where $M$ is a random variable whose law under $\pr_{X_t}$ is identical to the conditional law of $\sup_{[t,\zeta]}X_s$ under $\pr_x$ given $\mc{F}_t$ and $\{t<\zeta\}$.

Second, we identify $v(\cdot)$ as the expected payoff of the up-crossing strategy $T_{x^\star}^+$. That is,
\be
v(x)=\ex_{x}[\exp(-A_{T_{x^{\star}}^{+}})f(X_{T_{x^{\star}}^+})\ind_{\{T_{x^{\star}}^{+}<\infty\}}].\label{eq:hitpay}
\ee
In fact,  for any $z\in\R$, by conditioning,
\begin{align*}
&\ex_{x}[h(\overline{X}_{\zeta})\ind_{\{\overline{X}_{\zeta}>z\}}]
=\ex_{x}[h(\overline{X}_{\zeta})\ind_{\{T_{z}^{+}<\zeta, T_{z}^+<\infty\}}]=\ex_{x}[h(\overline{X}_{\zeta})\ind_{\{A_{T_{z}^{+}}<\mathbf{e}\}}\ind_{\{T_{z}^+<\infty\}}]\\
=&\ex_{x}[\ind_{\{T_{z}^+<\infty\}}\ind_{\{A_{T_{z}^{+}}<\mathbf{e}\}}\ex_{x}[h(\overline{X}_{\zeta})|\mathcal{F}_{T_z^+},T_z^+<\zeta]].
\end{align*}
On the event $\{T_z^+<\zeta\}$, the identity
$\ol{X}_\zeta=\ol{X}_{T_z^+}\vee\sup_{t\in[T_z^+,\zeta]}X_t=\sup_{t\in[T_z^+,\zeta]}X_t$ holds almost surely because $\ol{X}_{T_z^+}=X_{T_z^+}$, hence we have
\begin{align}
\ex_{x}[h(\overline{X}_{\zeta})|\mathcal{F}_{T_z^+},T_z^+<\zeta]
=\ex_{X_{T_z^+}}[h(M)]=f(X_{T_z^+}),\nn
\end{align}
where $M$ is a random variable whose law under $\pr_{X_{T_z^+}}$ is identical to the conditional law of $\sup_{t\in[T_z^+,\zeta]}X_t$ under $\pr_x$ given $\mc{F}_{T_z^+}$ and $\{T_z^+<\zeta\}$, from which the last equality results.

It follows from the tower property of conditional expectations that
\begin{align}\ex_{x}[h(\overline{X}_{\zeta})\ind_{\{\overline{X}_{\zeta}>z\}}]=\ex_x[\ex_x[\ind_{\{T_z^+<\infty\}}\ind_{\{A_{T_z^+}<\mathbf{e}\}}f(X_{T_z^+})|\mathcal{F}_{T_z^+}]]=\ex_x[\exp(-A_{T_z^+})f(X_{T_z^+})\ind_{\{T_z^+<\infty\}}].\label{eq:hitpay1}\end{align}
Applying \eqref{eq:hitpay1} for $z=x^\star$ we obtain \eqref{eq:hitpay}.

As a consequence, we know that, for any $x\ge x^\star$, $v(x)=\ex_{x}[h(\overline{X}_{\zeta})]=f(x)$. And it is easy to see that
\begin{align}
f(x)-v(x)=\ex_x[h(\ol{X}_\zeta)]-\ex_x[h(\ol{X}_\zeta)\ind_{\{\ol{X}_\zeta>x^\star\}}]=\ex_x[h(\ol{X}_\zeta)\ind_{\{\ol{X}_\zeta\leq x^\star\}}]\leq0,\nn
\end{align}
since $h(\ol{X}_\zeta)\ind_{\{\ol{X}_\zeta\leq x^\star\}}$ is a non-positive random variable.

In summary,  we have proved that, for any $x\in\R$, $(\et^{-A_t}v(X_t))_{t\ge0}$ is a positive $\pr_x$-supermartingale, and that $v(x)\ge \max\{f(x),0\}
\ge f(x)$ for all $x\in\R$. To finish the proof, we need to show that, for any stopping time $\tau\in\mc{T}$, we have
\be
v(x)\ge \ex_x[\et^{-A_{\tau}}f(X_{\tau})\ind_{\{\tau<\infty\}}], \quad\forall x\in\R.\label{eq:objective}
\ee
To establish the above inequality (without knowing if $(\et^{-A_t}v(X_t))_{t\ge0}$ is c\`adl\`ag), we let $n\in\mathbb{N}$ be a positive integer, and define
\be
\tau_n:=\min\{\frac{m}{n}: m\in\mathbb{N}, \frac{m}{n}\ge\tau\}\text{ whenever }\tau<\infty.\nn
\ee
Then by applying the optional sampling theorem and Fatou's lemma to the discrete-time supermartingale $(\exp(-A_{\frac{k}{n}})v(X_\frac{k}{n}))_{k=0,1,\ldots}$, we obtain that
\begin{align}v(x)\ge &\ex_x[\et^{-A_{\tau_n}}v(X_{\tau_n})\ind_{\{\tau_n<\infty\}}]\ge\ex_x[\et^{-A_{\tau_n}}\max\{f(X_{\tau_n}),0\}\ind_{\{\tau_n<\infty\}}], \quad\forall x\in\R,\label{eq:optional1}\end{align}
where the last inequality follows from the fact that $v(x)\ge \max\{f(x),0\}$ for all $x\in\R$. On the other hand, notice  that $\{\tau_n<\infty\}=\{\tau<\infty\}$ for all $n\in\mathbb{N}$, and as $n\to\infty$, we have $\tau_n\downarrow\tau$ on $\{\tau<\infty\}$. Because L\'evy process $X_\cdot$ has c\`adl\`ag path almost surely, and $A_\cdot$ is continuous and nondecreasing, we know that, as $n\to\infty$, $X_{\tau_n}\to  X_\tau$ and $A_{\tau_n}\downarrow A_{\tau}$ on event $\{\tau<\infty\}$.  Moreover, we also know that the function $\max\{f(x),0\}$ is lower semi-continuous, so on event $\{\tau<\infty\}$, we have $\liminf_{n\to\infty}\max\{f(X_{\tau_n}),0\}\ge \max\{f(X_\tau),0\}$. Together, we apply Fatou's lemma and a property of liminf (see Lemma \ref{lem:inf}) to obtain that 
\begin{align}
\liminf_{n\to\infty}\ex_x[\et^{-A_{\tau_n}}\max\{f(X_{\tau_n}),0\}\ind_{\{\tau_n<\infty\}}]&\ge \ex_x[\liminf_{n\to\infty}(\et^{-A_{\tau_n}}\max\{f(X_{\tau_n}),0\})\ind_{\{\tau<\infty\}}]\nn\\
&= \ex_x[\lim_{n\to\infty}\et^{-A_{\tau_n}}\cdot\liminf_{n\to\infty}\max\{f(X_{\tau_n}),0\}\cdot\ind_{\{\tau<\infty\}}]\nn\\
&\ge\ex_x[\et^{-A_\tau}\max\{f(X_{\tau}),0\}\ind_{\{\tau<\infty\}}]\nn\\
&\ge\ex_x[\et^{-A_\tau}f(X_{\tau})\ind_{\{\tau<\infty\}}],\label{eq:optional2}
\end{align}
where the last inequality follows from the fact that $\max\{f(x),0\}\ge f(x)$ for all $x\in\R$.
 By combining \eqref{eq:optional1} with \eqref{eq:optional2}, we obtain \eqref{eq:objective}.
 This completes the proof.
\end{proof}

\begin{remark}\label{rmk:33}
A close look at the proof reveals that the lower semi-continuity is only used in the second inequality of \eqref{eq:optional2}. More generally, the final result still holds if the reward function $f(\cdot)$ is such that $\liminf_{n\to\infty}f(X_{\tau_n})\ge f(X_\tau)$ on event $\{\tau<\infty\}$, for any nonincreasing sequence of stopping times $\{\tau_n\}_{n\ge1}$ that converges to $\tau$ almost surely.
\end{remark}

\begin{remark}\label{rem32}
From the above proof, it is easily seen that the result of Theorem \ref{singlesolution} still holds if we drop the monotonicity condition for $h(\cdot)$ over $(-\infty,x^\star]$ in Definition \ref{conM0}. That is, for single stopping problems, the only condition needed for $h(\cdot)$ is that $h(x)>0$ if and only if $x>x^\star$ and $h(\cdot)$ is nondecreasing over $(x^\star,\infty)$. Moreover, 
it is easily seen that the optimal stopping region for problem \eqref{singleproblem} is  contained in the set $\mathop{Supp}^+(f):=\{x\in\R: f(x)>0\}$ provided it is nonempty, so
if there is an $x_0\in\mathbb{R}$ such that $\mathop{Supp}^+(f)$ is a nonempty subset of   $[x_0,\infty)$, then it suffices to apply Theorem  \ref{singlesolution} (and verify Condition (M)) only for $x>x_0$. For example, if the reward function is $(\et^x-K)^+$ with $K>0$, then we can take $x_0=\log K$ and verify Condition (M) for $\et^x-K$.
\end{remark}

\subsection{The case of spectrally negative L\'evy processes}\label{sec:SNLPs}
For a given reward function $f(\cdot)$, we have seen that a sufficient condition for the optimal strategy to be of threshold type is to look for a representation as prescribed in Theorem \ref{singlesolution}. In general, it is a very challenging, if not impossible task to obtain such  representation.
In this section, we focus on the special case that $X_\cdot$ is a spectrally negative L\'{e}vy process, and derive the $h$ function based on conditions on $f$ and the random discounting term $A_\cdot$.

\begin{remark}\label{rmk:hazard}
If $X_\cdot$ has no positive jumps, then from the strong Markov property and the additive property of $X_\cdot$ and $A_\cdot$, we notice that, for any $z_1>z_2\ge x$,
\begin{align}
\pr_x(\overline{X}_\zeta>z_1)=&\ex_x[\exp(-A_{T_{z_1}^+})\ind_{\{T_{z_1}^+<\infty\}}]\nn\\
=&\ex_x[\exp(-A_{T_{z_2}^+})\ind_{\{T_{z_2}^+<\infty\}}\ex_{z_2}[\exp(-A_{T_{z_1}^+})\ind_{\{T_{z_1}^+<\infty\}}]]\nn\\
=&\pr_x(\overline{X}_\zeta>z_2)\pr_{z_2}(\overline{X}_\zeta>z_1).\label{eq:law}
\end{align}
\end{remark}
Hence, the law of $\overline{X}_\zeta$ exhibits some properties similar as exponential random variable,
in the sense that
\[\pr_x(\ol{X}_\zeta>z_1|\ol{X}_\zeta>z_2)=\pr_{z_2}(\ol{X}_\zeta>z_1).\]
Moreover,  it becomes clear after \eqref{eq:law} that the cumulative distribution function of $\overline{X}_\zeta$ under $\pr_x$ will be differentiable at $z$  for every $x\le z$, so long as it holds for $x=z$. In other words, the ``hazard rate'' of $\overline{X}_\zeta$ under $\pr_x$, if exists, will not depend on the starting point $x$. This gives rise to the following assumption.
\begin{Assumption}\label{assumption'}
There exists a positive function $\Lambda(\cdot)\in\mathbf{L}^\ind_{\textup{loc}}(\R)$, such that,
\begin{align}\label{eq:A321}
\int_{x}^\infty\Lambda(z)\diff z=\infty,\quad
\Lambda(z)=-\frac{1}{\pr_x(\overline{X}_\zeta>z)}\frac{\diff\pr_x(\overline{X}_\zeta>z)}{\diff z},\quad\,\,\forall x<&z.
\end{align}
\end{Assumption}
\begin{remark}\label{rmk:constantr}
In the case that the discount factor rate is a constant $r>0$,  the running maximum $\ol{X}_\zeta=\ol{X}_{\mathbf{e}_r}$ follows an exponential distribution with mean $1/\Phi(r)$, where $\mathbf{e}_r$ is an independent exponential random variable with mean $1/r$, and $\Phi(r)>0$ is the right inverse of the Laplace exponent of $X_\cdot$, see Appendix \ref{sec:SNLP}. So we have $\Lambda(z)=\Phi(r)$ for all $z\in\R$.
\end{remark}

\begin{cor}\label{cor1}
Under Assumption \ref{assumption'}, we have that
\be
\pr_x(\overline{X}_\zeta> z)=\exp(-\int_x^z\Lambda(y)\diff y),\,\quad\forall z>x.\nn
\ee
In particular, \eqref{eq:A321} implies that $\pr_x(\ol{X}_\zeta<\infty)=1$ (so Assumption \ref{ass:A} holds by Lemma \ref{lem:equivalence}). Moreover, on the event $\{T_z^+<\infty\}$, we have
\be
\pr_x(\overline{X}_{\zeta}>y|\mathcal{F}_{T_z^+}, \overline{X}_{\zeta}>z)=\exp(-\int_{z}^y\Lambda(u)\diff u),\quad\forall y>z\ge x.\label{eq:cond}
\ee
\end{cor}
\begin{proof}
We only prove \eqref{eq:cond} below. Notice that $\{\overline{X}_\zeta>y\}=\{T_y^+<\zeta,T_y^+<\infty\}=\{A_{T_y^+}<\mathbf{e}, T_y^+<\infty\}$. By the memoryless property of exponential random variable, we have
\[\pr_x(A_{T_y^+}<\mathbf{e}|\mathcal{F}_\infty, A_{T_z^+}<\mathbf{e}, T_y^+<\infty)=\exp(A_{T_y^+}-A_{T_{z}^+})=\exp(A_{T_y^+}\circ\theta_{T_z^+}).\]
The conclusion now follows from iterated conditional expectations given $\mathcal{F}_{T_z^+}$.
\end{proof}

To find a representation of the reward function $f(\cdot)$ as in Definition \ref{conM0}, we make the following assumption.
\begin{Assumption}\label{assumption''}
The reward function $f(\cdot)$ satisfies the following:
\begin{enumerate}
\item[(i)] $f(x)>0$ for all sufficiently large $x$;
\item[(ii)] The reward function $f(\cdot)$ is absolutely continuous with respect to the Lebesgue measure. Let $h(x)=f(x)-\frac{f'(x)}{\Lambda(x)}$, a.e. $x\in\mathbb{R}$. Then there is an $x^\star\in[-\infty,\infty]$ such that
\begin{enumerate}
\item[(a)] $h(x)>0$ a.e. $x>x^\star$ and $h(x)\le 0$ a.e. $x\le x^\star$ (if $x^\star=\infty$ then $h(x)\le 0$ for all $x\in\R$);
\item[(b)] the function $h(\cdot)$ is nondecreasing over $(x^\star,\infty)$.\footnote{\label{fn:aeh}This means that there is a nondecreasing function $\tilde{h}(\cdot)$, such that, $h(x)=\tilde{h}(x)$, almost everywhere on $(x^\star,\infty)$.}
\end{enumerate}
\end{enumerate}
\end{Assumption}
\begin{remark}\label{rmk:constantr1}
By using \eqref{eq:p2e}, Corollary \ref{cor1} and Assumption \ref{assumption''}(ii), we know that, the mapping 
\[z\mapsto \ex_x[\et^{-A_{T_z^+}}f(X_{T_z^+})\ind_{\{T_z^+<\infty\}}]=f(z)\pr_x(\ol{X}_\zeta>z)=f(z)\exp(-\int_x^z\Lambda(y)\diff y),\quad\forall z>x,\] is differentiable almost everywhere, and that this function is nondecreasing over $[x,x^\star\vee x]$, and is strictly decreasing over $[x^\star\vee x,\infty)$, i.e., the function is maximized at $x^\star\vee x$ for each $x$.
In the case that the discount factor rate is a constant $r>0$, then for any fixed $\beta\in(0,\Phi(r))$ and $K>0$, functions $(e^{\beta x}+\et^{\Phi(r)x}-K)^+$ and $(\et^{\Phi(r)x}-K)^+$ satisfy Assumption \ref{assumption''}(i),(ii), with $x^\star=\frac{1}{\beta}\log(K/(1-\frac{\beta}{\Phi(r)}))$ and $x^\star=\infty$, respectively. 
\end{remark}

The following lemma explains why Assumption \ref{assumption''} is necessary for Condition (M) to hold.
\begin{prop}\label{prop:necessary}
Suppose Assumption \ref{assumption'} holds, $f(\cdot)$ satisfies Condition (M) and let $h(\cdot)$ be the nondecreasing function in Definition \ref{conM0}. Then all conditions in Assumption \ref{assumption''} hold. 
\end{prop}
\begin{proof}
By \eqref{eq:conM0} and Corollary \ref{cor1}, we know that 
\[f(x)=\ex_x[h(\ol{X}_\zeta)]=\int_{x}^\infty h(z)\Lambda(z)\et^{-\int_x^z\Lambda(y)\diff y}\diff z,\quad\forall x\in\R.\]
Clearly, $f(x)>0$ for all $x>x^\star$, so Assumption \ref{assumption''}(i) holds. On the other hand, from $f'(x)=-\Lambda(x)h(x)+\Lambda(x)f(x)$, we know that $f(\cdot)$ is absolutely continuous, and $h(x)=f(x)-\frac{f'(x)}{\Lambda(x)}$, a.e.. So  Assumption \ref{assumption''}(ii) holds.
\end{proof}

Before we prove a reverse of Proposition \ref{prop:necessary}, we show that the limit of the value for the up-crossing strategy $T_z^+$ as the threshold $z\to\infty$ is well-defined. 

\begin{lem}\label{lem:limit}
Under Assumption \ref{assumption'} and Assumption \ref{assumption''}, the limit 
 \be
c_0:=\lim_{z\to\infty}f(z)\et^{-\int_0^z\Lambda(y)\diff y}=\lim_{z\to\infty}\ex[\et^{-A_{T_z^+}}f(X_{T_z^+})\ind_{\{T_z^+<\infty\}}],\label{eq:defc0}
\ee
exists. If $x^\star<\infty$, we have $c_0\in[0,\infty)$; if $x^\star=\infty$, then we have $c_0\in(0,\infty]$.
\end{lem}

\begin{remark}\label{rmk:35}
If  Condition (M) holds for $f(\cdot)$ and $f(x)=\ex_x[h(\ol{X}_\zeta)]$, then, from \eqref{eq:hitpay1} we have
\begin{align*}f(z)\et^{-\int_x^z\Lambda(y)\diff y}=\ex_x[\exp(-A_{T_z^+})f(X_{T_z^+})\ind_{\{T_z^+<\infty\}}]=\ex_x[h(\ol{X}_\zeta)\ind_{\{\ol{X}_\zeta>z\}}],\quad\forall x<z,\end{align*}
 we know that the limit of the above as $z\to\infty$ is zero (i.e., $c_0=0$), because of the integrability of random variable $h(\ol{X}_\zeta)$. Therefore, given Assumption \ref{assumption'}, the conditions given in Assumption \ref{assumption''} are more general than those in Definition \ref{conM0}.

\end{remark}

\begin{prop}\label{lem}
Suppose that Assumptions \ref{assumption'} and Assumption \ref{assumption''}  hold, and $c_0$ defined in \eqref{eq:defc0} is finite.
Then it holds that
\be
\ex_{x}[|h(\overline{X}_{\zeta})|]<\infty,\quad f(x)-c_0 \et^{\int_0^x\Lambda(y)\diff y}=\ex_{x}[h(\overline{X}_{\zeta})],\quad\forall x\in\R.\nn
\ee
\end{prop}
\begin{proof}
We only prove the case that $x^\star\in\R$, the remaining cases that $x^\star$ is $\pm\infty$ can be proved similarly. 
To prove the finiteness of $\ex[|h(\overline{X}_\zeta)|]$, we prove that both $\ex[h(\overline{X}_\zeta)\ind_{\{\overline{X}_\zeta>x^\star\}}]<\infty$ and $-\ex[h(\overline{X}_\zeta)\ind_{\{\overline{X}_\zeta\le x^\star\}}]<\infty$ hold. 
To establish the former, we fix an $x$ and any constant $D>x^\star$, we use Corollary \ref{cor1} to obtain that
\begin{align*}
\ex_x[h(\overline{X}_\zeta)\ind_{\{D>\overline{X}_\zeta>x^\star\}}]=&\int_{x^\star}^D h(z)\pr_x(\overline{X}_\zeta\in\diff z)=\int_{x^\star\vee x}^{D\vee x} f(z)\Lambda(z)\et^{-\int_x^z\Lambda(y)\diff y}\diff z-\int_{x^\star\vee x}^{D\vee x} f'(z)\et^{-\int_x^z\Lambda(y)\diff y}\diff z\nn\\
=&f(x^\star\vee x)\et^{-\int_x^{x^\star\vee x}\Lambda(y)\diff y}-f(D\vee x)\et^{-\int_x^{D\vee x}\Lambda(y)\diff y},
\end{align*}
where the last step is due to integration by parts (see, e.g., \cite[Theorem 9 on page 163]{measure_integration81}).
Take the limit as $D\to\infty$, we use the monotone convergence theorem to obtain that (also using the a.s. finiteness of $\overline{X}_\zeta$)
\begin{align}
\ex_x[h(\overline{X}_\zeta)\ind_{\{\overline{X}_\zeta>x^\star\}}]&=\lim_{D\to\infty}\ex_x[h(\overline{X}_\zeta)\ind_{\{D>\overline{X}_\zeta>x^\star\}}]\nn\\
&=f(x^\star\vee x)\et^{-\int_x^{x^\star\vee x}\Lambda(y)\diff y}-\lim_{D\to\infty}f(D)\et^{-\int_x^D\Lambda(y)\diff y}\nn\\
&=f(x^\star\vee x)\et^{-\int_x^{x^\star\vee x}\Lambda(y)\diff y}-c_0 \et^{\int_0^x\Lambda(y)\diff y}<\infty,\label{eq:15}
\end{align}
where  the last step  is due to Lemma \ref{lem:limit}.

Similarly, for any fixed $x$,
\begin{align}
-\ex_x[h(\overline{X}_\zeta)\ind_{\{\overline{X}_\zeta\le x^\star\}}]&=-\int_x^{x^\star\vee x} h(z)\pr_x(\overline{X}_\zeta\in\diff z)\nn\\
&=\int_x^{x^\star\vee x} f'(z)\et^{-\int_x^z\Lambda(y)\diff y}\diff z-\int_x^{x^\star\vee x} f(z)\Lambda(z)\et^{-\int_x^z\Lambda(y)\diff y}\diff z\nn\\
&=f(x^\star\vee x)\et^{-\int_x^{x^\star\vee x}\Lambda(y)\diff y}-f(x)<\infty.\label{eq:16}
\end{align}
Therefore, we know that $\ex_x[|h(\overline{X}_\zeta)|]<\infty$. Moreover, from \eqref{eq:15} and \eqref{eq:16} we also obtain that
\begin{align*}
\ex_{x}[h(\overline{X}_{\zeta})]&=\ex_{x}[h(\overline{X}_{\zeta})\ind_{\{\overline{X}_\zeta>x^\star\}}]+\ex_{x}[h(\overline{X}_{\zeta})\ind_{\{\overline{X}_\zeta\le x^\star\}}]=f(x)-c_0 \et^{\int_0^x\Lambda(y)\diff y}.
\end{align*}
This completes the proof.\end{proof}

Below we present the main result of this section, a {\em generalization} of Theorem \ref{singlesolution}. 

\begin{thm}\label{thmSNLP}
Under Assumption \ref{assumption'} and Assumption \ref{assumption''},  we have
\begin{align}
V(x)=&\sup_{\tau\in\mc{T}}\ex_x[\et^{-A_\tau}f(X_\tau)\ind_{\{\tau<\infty\}}]=\ex_x[h(\ol{X}_\zeta)\ind_{\{\ol{X}_\zeta>x^\star\}}]+c_0 \et^{\int_0^x\Lambda(y)\diff y}.\end{align}
If $x^\star<\infty$, then the above value can be attained by stopping time $T_{x^\star}^+$. 
If $\Lambda(\cdot)$ is continuous and $c_0<\infty$,  and in the case of $x^\star<\infty$ we also have $f(\cdot)$ is continuously differentiable, then $V(\cdot)$ is also continuously differentiable. In particular, smooth fit holds at $x=x^\star$ if $x^\star\in(-\infty,\infty)$.
\end{thm}
\begin{proof}
 Let us first suppose that $x^\star<\infty$ and $c_0=0$. In this case, Theorem \ref{singlesolution}  already gives us both the value function and the optimal stopping time, and hence 
we only need to prove  that the smooth fit condition holds when $x^\star$ is finite. To that end, we notice that
\[V(x)=\int_{x^\star}^\infty h(z)\Lambda(z)e^{-\int_x^z\Lambda(y)\diff y}\diff z=e^{-\int_{x}^{x^\star}\Lambda(y)\diff y}\int_{x^\star}^\infty h(z)\Lambda(z)e^{-\int_{x^\star}^z\Lambda(y)\diff y}\diff z, \quad\forall x\le x^\star,\]
which is continuously differentiable over $(-\infty,x^\star)$. Indeed,
\[V'(x)=\Lambda(x)V(x),\]
so we have
$V'(x^\star-)=\Lambda(x^\star)V(x^\star)=\Lambda(x^\star)f(x^\star).$
Thus,
\[V'(x^\star-)-V'(x^\star+)=V'(x^\star-)-f'(x^\star)=\Lambda(x^\star)\bigg(f(x^\star)-\frac{f'(x^\star)}{\Lambda(x^\star)}\bigg)=\Lambda(x^\star)h(x^\star)=0,\]
which indicates that $V(\cdot)$ satisfies smooth fit at $x=x^\star$ if $x^\star\in(-\infty,\infty)$.

Let us suppose $x^\star<\infty$ and $c_0>0$. In this case, we can apply Theorem \ref{singlesolution} to the reward function $f(x)-c_0\et^{\int_0^x\Lambda(y)\diff y}$, to obtain that 
\begin{align}
\ex_x[h(\ol{X}_\zeta)\ind_{\{\ol{X}>x^\star\}}]&=
\sup_{\tau\in\mc{T}}\ex_x[\et^{-A_\tau}(f(X_\tau)-c_0\et^{\int_0^{X_\tau}\Lambda(y)\diff y})\ind_{\{\tau<\infty\}}]\nn\\
&=\ex_x[\et^{-A_{T_{x^\star}^+}}(f(X_{T_{x^\star}^+})-c_0\exp(\int_0^{X_{T_{x^\star}^+}}\Lambda(y)\diff y))\ind_{\{{T_{x^\star}^+}<\infty\}}]\nn\\
&=\ex_x[\et^{-A_{T_{x^\star}^+}}f(X_{T_{x^\star}^+})\ind_{\{{T_{x^\star}^+}<\infty\}}]-c_0\et^{\int_0^{x\vee x^\star}\Lambda(y)\diff y}\ex_x[\et^{-A_{T_{x^\star}^+}}\ind_{\{{T_{x^\star}^+}<\infty\}}]\nn\\
&=\ex_x[\et^{-A_{T_{x^\star}^+}}f(X_{T_{x^\star}^+})\ind_{\{{T_{x^\star}^+}<\infty\}}]-c_0\et^{\int_0^x\Lambda(y)\diff y}.\label{eq:c0step1}
\end{align}
On the other hand, Lemma \ref{lem:Ulocal} in Appendix \ref{sec:proof} proves that  $(\exp(-A_t+\int_0^{X_t}\Lambda(y)\diff y))_{t\ge0}$ is a nonnegative c\`adl\`ag local martingale, hence it is a supermartingale.  Hence, for any stopping time $\tau\in\mc{T}$, by using the optional sampling theorem and Fatou's lemma, we have
\[0\le \ex_x[\et^{-A_\tau+\int_0^{X_{\tau}}\Lambda(y)\diff y}\ind_{\{\tau<\infty\}}]\le \et^{\int_0^x\Lambda(y)\diff y}.\]
So 
\be
0\le \sup_{\tau\in\mc{T}}\ex_x[\et^{-A_\tau}\cdot c_0\et^{\int_0^{X_\tau}\Lambda(y)\diff y}\ind_{\{\tau<\infty\}}]\le \et^{\int_0^x\Lambda(y)\diff y}.\label{eq:c0step2}
\ee
By the well-known properties of supremum, we obtain from \eqref{eq:c0step1} and \eqref{eq:c0step2} that
\begin{align}
&\sup_{\tau\in\mc{T}}\ex_x[\et^{-A_\tau}\cdot f(X_\tau)\ind_{\{\tau<\infty\}}]\nn\\
\le\,&\sup_{\tau\in\mc{T}}\ex_x[\et^{-A_\tau}(f(X_\tau)-c_0\et^{\int_0^{X_\tau}\Lambda(y)\diff y})\ind_{\{\tau<\infty\}}]+\sup_{\tau\in\mc{T}}\ex_x[\et^{-A_\tau}\cdot c_0\et^{\int_0^{X_\tau}\Lambda(y)\diff y}\ind_{\{\tau<\infty\}}]\nn\\
\le \,&\ex_x[h(\ol{X}_\zeta)\ind_{\{\ol{X}>x^\star\}}]+c_0\et^{\int_0^x\Lambda(y)\diff y}=\ex_x[\et^{-A_{T_{x^\star}^+}}f(X_{T_{x^\star}^+})\ind_{\{{T_{x^\star}^+}<\infty\}}].\nn
\end{align}
Thus, all inequalities in the above are in fact equalities. This proves the optimality of $T_{x^\star}^+$. Finally, smooth fit at $x^\star$ when $x^\star$ is finite can be proved similarly as before. 

If $x^\star=\infty$ and $c_0\in(0,\infty)$, then by footnote \ref{fn:xinfty} and Proposition \ref{lem}, we know that 
\[0=\ex_x[h(\ol{X}_\zeta)\ind_{\{\ol{X}>x^\star\}}]=
\sup_{\tau\in\mc{T}}\ex_x[\et^{-A_\tau}(f(X_\tau)-c_0\et^{\int_0^{X_\tau}\Lambda(y)\diff y})\ind_{\{\tau<\infty\}}].\]
Using the same argument as above, we have 
\be
\sup_{\tau\in\mc{T}}\ex_x[\et^{-A_\tau}f(X_\tau)\ind_{\{\tau<\infty\}}]\le c_0\et^{\int_0^x\Lambda(y)\diff y}.\label{eq:c0step3}
\ee
On the other hand, we trivially have (for $z>x$)
\be
\sup_{\tau\in\mc{T}}\ex_x[\et^{-A_\tau}f(X_\tau)\ind_{\{\tau<\infty\}}]\ge \ex_x[\et^{-A_{T_{z}^+}}f(X_{T_{z}^+})\ind_{\{{T_{z}^+}<\infty\}}]=f(z)\et^{-\int_x^z\Lambda(y)\diff y}\to c_0\et^{\int_0^x\Lambda(y)\diff y},\label{eq:c0step4}
\ee
as $z\to\infty$, thanks to Lemma \ref{lem:limit}. It follows that the inequality in \eqref{eq:c0step3} is an equality. The value function is clearly continuously differentiable. 

If $x^\star=c_0=\infty$, then by \eqref{eq:c0step4} we know that the optimal value is $\infty$.
\end{proof}

\sk
\sk
From the proof of Theorem \ref{thmSNLP}, we immediately obtain the following result. 
\begin{cor}\label{cor:mart}
The positive process  $(\exp(-A_t+\int_0^{X_t}\Lambda(y)\diff y))_{t\ge0}$ is a true martingale.
\end{cor}

Finally, we show that Theorem \ref{thmSNLP} implies the results in \cite[Theorem 3.1]{HLY14} on optimality of a threshold type strategy in optimal stopping problem with constant discounting rate $r>0$ and spectrally negative L\'evy model.

\begin{cor}\label{cor:snlp}
Assuming that the reward function $f(\cdot)$ is log-concave, increasing and non-negative. 
Define function
\be
h(x):=f(x)-\frac{f'(x-)}{\Phi(r)}. \nn
\ee
Then there is a constant $x^\star\in[-\infty,\infty]$ such that $h(x)>0$ if and only if $x>x^\star$, and $h(\cdot)$ is nondecreasing over $(x^\star,\infty)$. Moreover, if $x^\star<\infty$, 
\be
V(x)=\sup_{\tau\in\mc{T}}\ex_x[\et^{-r\tau}f(X_\tau)\ind_{\{\tau<\infty\}}]=\ex_x[h(\ol{X}_{\e_r})\ind_{\{\ol{X}_{\e_r>x^\star}\}}]=\ex_x[\et^{-rT_{x^\star}^+}f(X_{T_{x^\star}^+})\ind_{\{T_{x^\star}^+<\infty\}}].\nn
\ee
If $x^\star=\infty$, then the value function is given by 
\[V(x)=\sup_{\tau\in\mc{T}}\ex_x[\et^{-r\tau}f(X_\tau)\ind_{\{\tau<\infty\}}]=\lim_{z\to\infty}\ex_x[\et^{-rT_z^+}f(X_{T_z^+})\ind_{\{T_z^+<\infty\}}]=c_0\cdot\et^{\Phi(r)x},\]
where $c_0$ is defined in \eqref{eq:defc0}.
\end{cor}
\begin{proof}
First, Assumption \ref{assumption'}  and Assumption \ref{assumption''}(i) obviously hold. Therefore, to apply Theorem \ref{thmSNLP}, we only need to verify that Assumption \ref{assumption''}(ii) holds as well.  

Let $\mathrm{Supp}(f):=\{x\in\R:f(x)>0\}$ be the support of the reward function $f(\cdot)$. Since $\log f(\cdot)$ is concave on $\mathrm{Supp}(f)$,  we know that $f(\cdot)$ is absolutely continuous, and  the left-hand derivative $\big(\log f(x-)\big)'=f'(x-)/f(x)$ is non-increasing over $\mathrm{Supp}(f)$.
On the other hand,
\be
h(x)>0\Leftrightarrow f(x)-\frac{f'(x-)}{\Phi(r)}>0\Leftrightarrow \frac{f'(x-)}{f(x)}<\Phi(r)\Leftrightarrow\big(\log f(x-)\big)'<\Phi(r),\label{eq:cor36}
\ee
where the second step comes from the observation that $\{x\in\R:h(x)>0\}\subset\mathrm{Supp}(f)$, thanks to $f'(x-)\ge0$. By the monotonicity of $\big(\log f(x-)\big)'$,  let us define
 $$x^\star:=\inf\{x\in\R:\big(\log f(x-)\big)'<\Phi(r)\}\in[-\infty,\infty].$$

If $x^\star=\infty$, then $h(x)\le 0$ for all $x\in\R$ so there is nothing more to verify. Below we assume that $x^\star<\infty$, and 
prove that $h(\cdot)$ is nondecreasing over $(x^\star,\infty)$. To this end, consider
 any $x^\star<x\leq y$, then we have
\begin{align}
f'(y-)-f'(x-)=&\big(\log f(y-)\big)'f(y)-\big(\log f(x-)\big)'f(x)\nn\\
\leq& \big[\big(\log f(y-)\big)'-\big(\log f(x-)\big)'\big]f(y)+\big(\log f(x-)\big)'[f(y)-f(x)]\nn\\
\leq& \big(\log f(x-)\big)'[f(y)-f(x)],\nn
\end{align}
where the last inequality holds since $\log f(\cdot)$ is concave and $\big(\log f(\cdot-)\big)'$ is non-increasing. However, we know that $f(y)\ge f(x)$, and for $x>x^\star$ we have $\big(\log f(x-)\big)'\leq\Phi(r)$. Thus, we have
\be
f'(y-)-f'(x-)\leq \Phi(r)(f(y)-f(x)),\nn
\ee
which holds if and only if
\be
h(y)-h(x)=\bigg(f(y)-\frac{f'(y-)}{\Phi(r)}\bigg)-\bigg(f(x)-\frac{f'(x-)}{\Phi(r)}\bigg)\geq0.\nn
\ee
Therefore $h(\cdot)$ is nondecreasing on $(x^\star,\infty)$, so Assumption \ref{assumption''}(ii) holds.
\end{proof}
\subsection{Examples}\label{sec:Example}

\subsubsection{Discounting with local time}
We consider a generalization of the local time discounting problem as studied in \cite{Dayanik2008b}. More specifically, we let $X_\cdot$ be a spectrally negative $\beta$-stable process with index $\beta\in(1,2]$, and $L_t$ be the local time of $X_\cdot$ at level 0, which is defined as the occupation time density at $0$. That is, 
\be
L_t=\lim_{\epsilon\downarrow 0}\frac{1}{2\epsilon}\int_0^t\ind_{(-\epsilon,\epsilon)}(X_s)\diff s,\quad\pr\text{-a.s.}\label{eq:LocalT}\ee
Thanks to the fact that 0 is regular for itself when $X_\cdot$ has unbounded variation (see, e.g., \cite[Corollary VII.5]{Bertoin_1996}), we know from \cite[page 327]{Sato99} that the occupation density $L_t$ defined in \eqref{eq:LocalT} exists. 

Fix constants $r,\alpha>0$, our objective is to solve the following optimal stopping problem:
\be
V(x):=\sup_{\tau\in\mathcal{T}}\ex_{x}[\et^{-r\,L_{\tau}}(X_{\tau}\vee0)^\alpha\ind_{\{\tau<\infty\}}].\label{problem:local}
\ee
 This problem was studied in the special case of  $\beta=2$, i.e., standard Brownian motion, in \cite[Section 4.2]{Dayanik2008b}. Our objective is to extend this result to a general $\beta\in(1,2]$. 

We begin by deriving the law of the local time stopped at the up-crossing strategy $T_z^+$, $z>0$.
\begin{lem}\label{lem1}
Let $X_\cdot$ be a general spectrally negative L\'evy process with unbounded variation (so that $W(0)=0$), we have
\be
\ex_x[\et^{-r\,L_{T_z^+}}\ind_{\{T_z^+<\infty\}}]=\frac{\et^{\Phi(0)x}+rW(x)}{\et^{\Phi(0)z}+rW(z)}, \quad\forall z>x\ge0,\label{eq:lem}
\ee
where $W(\cdot)$ is the 0-scale function of $X_\cdot$ (see Appendix \ref{sec:SNLP}).
\end{lem}

In the special case that  $X_\cdot$ is a spectrally negative $\beta$-stable process, it is well-known that $\Phi(0)=0$ and hence $\pr_x(T_z^+<\infty)=1$ for any $z>0$. Moreover, from \cite[page 233]{Kyprianou2006} we know that
\begin{align}
W(x)=\frac{x^{\beta-1}}{\Gamma(\beta)}\ind_{\{x\ge 0\}},\quad\forall x\in\mathbb{R}.\nn
\end{align}
Therefore, from  Lemma \ref{lem1}  we know that, for any $z>x\ge 0$,
\[\pr_x(\ol{X}_\zeta>z)=\ex_x[\et^{-rL_{T_z^+}}\ind_{\{T_z^+<\infty\}}]=\frac{\Gamma(\beta)+r{x}^{\beta-1}}{\Gamma(\beta)+rz^{\beta-1}},\]
where $\zeta=\inf\{t>0: rL_t>\mathbf{e}\}$. Using the limit of the above equation as $z\to\infty$, one also knows that the additive functional $L$ satisfies Assumption \ref{ass:A}, thanks to Lemma \ref{lem:equivalence}. On the other hand, by  the strong Markov Property of $X_\cdot$ at stopping time $T_0^+$, we have for any $x<0$ that,
\[\pr_x(\ol{X}_\zeta>z)=\ex_x[\et^{-rL_{T_z^+}}\ind_{\{T_z^+<\infty\}}]=\ex_x[\ex_{X_{T_0^+}}[\et^{-rL_{T_z^+}}\ind_{\{T_z^+<\infty\}}]]=\ex[\et^{-rL_{T_z^+}}\ind_{\{T_z^+<\infty\}}]=\frac{\Gamma(\beta)}{\Gamma(\beta)+rz^{\beta-1}}.\]
It follows that the ``hazard rate'' of  $\ol{X}_\zeta$ is given by
\begin{align*}
\Lambda{(z)}=&-\frac{1}{\pr_x(\ol{X}_\zeta>z)}\frac{\partial}{\partial z}\pr_x(\ol{X}_\zeta>z)=\frac{r(\beta-1)z^{\beta-2}}{\Gamma(\beta)+rz^{\beta-1}},\quad \forall z>x>0.
\end{align*}
So  we know that   Assumption \ref{assumption'} holds. Obviously, Assumption \ref{assumption''}(i) also holds. 

To verify Assumption \ref{assumption''}(ii), we consider the reward function $f(x)=x^\alpha\ind_{\{x>0\}}$, and define for any $x>0$,
\begin{align}
h(x):=&x^\alpha-\frac{\alpha x^{\alpha-1}}{\Lambda(x)}=x^\alpha-\alpha x^{\alpha-1}\frac{\Gamma(\beta)+rx^{\beta-1}}{r(\beta-1)x^{\beta-2}}=\frac{\beta-\alpha-1}{\beta-1}x^\alpha-\Gamma(\beta-1)\frac{\alpha}{r} x^{\alpha-\beta+1}.\nn
\end{align}
Let us first assume that $\alpha\in(0,\beta-1)$. Then it is easily seen that $h(\cdot)$ satisfies the following properties (notice that $\beta-\alpha-1\ge0,\alpha>0$):
\begin{align}
h'(x)&=\frac{\alpha(\beta-\alpha-1)}{\beta-1}x^{\alpha-1}+(\beta-\alpha-1)\Gamma(\beta-1)\frac{\alpha}{r} x^{\alpha-\beta}>0,\nn\\
h(x)&> 0 \text{ if and only if }x>x^\star\equiv \bigg(\frac{\alpha\Gamma(\beta)}{r(\beta-\alpha-1)}\bigg)^{\frac{1}{\beta-1}}.\label{eq:xstar_alpha}
\end{align}
Hence, when $\alpha\in(0,\beta-1)$, Assumption \ref{assumption''}(ii) holds with a finite $x^\star$. Furthermore, we compute $c_0=\lim_{z\to\infty}z^\alpha\ex[\et^{-rL_{T_z^+}}\ind_{\{T_z^+<\infty\}}]=0.$

If $\alpha=\beta-1$, then we have $h(x)=-\Gamma(\beta)/r<0$ for all $x>0$. Hence Assumption \ref{assumption''}(ii) holds with $x^\star=\infty$. 
Moreover, we compute $c_0=\lim_{z\to\infty}z^\alpha\ex[\et^{-rL_{T_z^+}}\ind_{\{T_z^+<\infty\}}]=\frac{\Gamma(\beta)}{r}$.

If $\alpha>\beta-1$, then from $h(0)=0$ and $h'(x)<0$ for all $x>0$, we know that $h(x)<0$ for all $x>0.$ Hence Assumption \ref{assumption''}(ii) holds with $x^\star=\infty$. Moreover, we compute $c_0=\lim_{z\to\infty}z^\alpha\ex[\et^{-rL_{T_z^+}}\ind_{\{T_z^+<\infty\}}]=\infty$.

By Remark \ref{rem32} and Theorem \ref{thmSNLP}, we obtain the following result:

\begin{figure}
\centering
\includegraphics[width=3.5in]{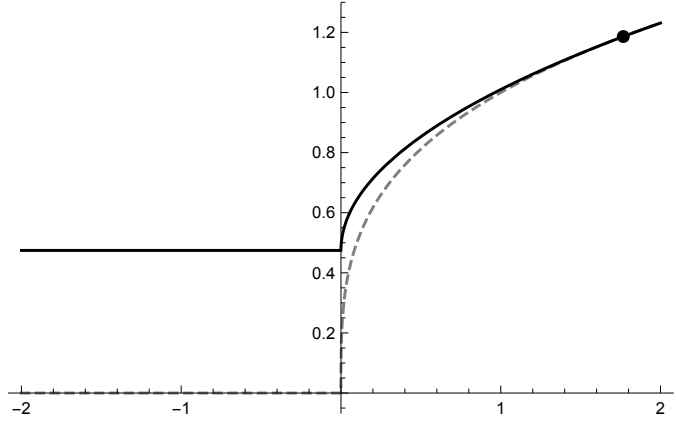}
\caption{Here we plot the reward function $f(x)=(x\vee0)^\alpha$ (in gray dashed line) and the value function given in \eqref{eq:prop32} (in black solid line).  Model parameters: $\alpha=0.3, \beta=1.5, r=1$. The optimal threshold $x^\star=1.7672$. According to Theorem \ref{thmSNLP}, smooth fit holds at $x^\star$. The black dot stands for the pasting point where optimal exercising occurs. }
\label{fig1}
\end{figure}

\begin{prop}\label{prop:localtime33}
If $\alpha\in(0,\beta-1)$, the optimal stopping time for problem \eqref{problem:local} is the up-crossing strategy $T_{x^\star}^+$. Moreover, (see Figure \ref{fig1} for a visualization)
\begin{align}
V(x)
&=\ex_x[\exp(-r\,L_{T_{x^\star}^+})(X_{T_{x^\star}^+})^\alpha\ind_{\{T_{x^\star}^+<\infty\}}]=\begin{dcases}
(x^\star)^\alpha\frac{\Gamma(\beta)+r(x\vee0)^{\beta-1}}{\Gamma(\beta)+r(x^\star)^{\beta-1}}, &\text{ if $x\le x^\star$},\\
(x)^\alpha, &\text{ if $x> x^\star$},
\end{dcases}\label{eq:prop32}
\end{align}
where $x^\star$ is defined in \eqref{eq:xstar_alpha}.

If $\alpha=\beta-1$, then the value function for problem \eqref{problem:local} is given by
\be
V(x)=\frac{\Gamma(\beta)}{r}+(x\vee 0)^{\beta-1},\quad\forall x\in\R.
\ee
If $\alpha>\beta-1$, then the value function for problem \eqref{problem:local} is $\infty$.
\end{prop}

\subsubsection{Discounting with occupation time}
We let $X_\cdot$ be any spectrally negative L\'evy process with L\'evy triplet $(\mu,\sigma,\Pi)$, such that the tail jump measure $\Pi(-\infty,-x)$ has a completely monotone density over $(0,\infty)$. For fixed $r,q>0$, we consider the following optimal stopping problem:
\be
V(x):=\sup_{\tau\in\mathcal{T}}\ex_{x}[\et^{-r\tau-q\int_{0}^{\tau}\mathbf{1}_{(-\infty,0)}(X_{s})\diff s}(X_\tau\vee0)\,\ind_{\{\tau<\infty\}}],\quad \forall x\in\mathbb{R}.\label{problem:occupation}
\ee
 The occupation time process obviously satisfies Assumption \ref{ass:A}.
Notice that problem \eqref{problem:occupation} corresponds to a random discounting generalization of the Novikov-Shiryaev optimal stopping problem for degree one (see, e.g.,  \cite[Section 9.4]{Kyprianou2006} and \cite{KyprSury05}, see also \cite{NovkShir07} for a study of the same problem in discrete time). In the spirit of \cite{Linetsky99}, we can consider the problem as the evaluation of a perpetual American style step option  on risky asset $S_\cdot=(\et^{X_t})_{t\ge0}$ with reward  $\exp(-q\int_0^\tau\ind_{(0,1)}(S_s)\diff s) (\log S_\tau\vee0)$.

We begin by deriving the law of the occupation time stopped at an up-crossing stop time.
By  \cite[Proposition 4.1]{OmegaRZ15} (or \cite[Corollary 2(ii)]{OccupationInterval}), we know that for $x\leq z$,
\begin{align}
\ex_{x}[\et^{-rT_{z}^+-q\int_{0}^{T_{z}^+}\ind_{(-\infty,0)}(X_{s})\diff s}\ind_{\{T_{z}^+<\infty\}}]
=&\frac{\int_0^{\infty} \et^{-\Phi(r+q)y}W^{(r)}(x+y)\diff y}{\int_0^{\infty} \et^{-\Phi(r+q)y}W^{(r)}(z+y)\diff y},\nn
\end{align}
where $W^{(r)}(\cdot)$ is the $r$-scale function of $X$ (see Appendix \ref{sec:SNLP}).
It follows that, for the random variable $\ol{X}_\zeta$ with $\zeta=\inf\{t>0: rt+q\int_0^t\ind_{\{(-\infty,0)\}}(X_s)\diff s>\mathbf{e}\}$, we have
\[\pr_x(\ol{X}_\zeta>z)=\frac{\int_0^{\infty} \et^{-\Phi(r+q)y}W^{(r)}(x+y)\diff y}{\int_0^{\infty} \et^{-\Phi(r+q)y}W^{(r)}(z+y)\diff y},\quad\forall z>x,\]
and its
 ``Hazard rate'' $\Lambda$  given by \cite[Eq. (4.5)]{OmegaRZ15} becomes
\begin{align}
\Lambda(z)=-\frac{1}{\pr_x(\ol{X}_\zeta>z)}\frac{\p}{\p z}\pr_x(\ol{X}_\zeta>z)=\Phi(r+q)-\frac{W^{(r)}(z)}{\int_0^{\infty} \et^{-\Phi(r+q)y}W^{(r)}(z+y)\diff y},\quad \forall z>x.\label{eq:lambda}
\end{align}
In  \cite[Lemma 4.2]{OmegaRZ15}, it is proved that the function $\Lambda(z)$ is non-increasing over $\mathbb{R}$, satisfying
\[\Lambda(-\infty)=\Lambda(0-)=\Phi(r+q)\ge \Phi(r+q)-q W^{(r)}(0)=\Lambda(0)>\Lambda(\infty)=\Phi(r)>0.\]
 Clearly, Assumption \ref{assumption'} and Assumption \ref{assumption''}(i) hold.   To verify that Assumption \ref{assumption''}(ii) also holds, define
\be
h(x)=x-\frac{1}{\Lambda(x)},\quad\forall x>0,\nn
\ee
where we only consider positive $x$ because of Remark \ref{rem32}.
It is easily seen that $h(x)<0$ for sufficiently small $x>0$ and $h(x)\to\infty$ as $x\to\infty$. Moreover,
straightforward calculation yields that
\begin{align*}
h'(x)=\frac{1}{\Lambda(x)^2}\left(\Lambda'(x)+(\Lambda(x))^2\right)
=&\frac{(\Lambda(x))^{-2}}{\int_x^\infty\et^{-\Phi(r+q)y}W^{(r)}(y)\diff y}N(x),
\end{align*}
where we defined for any $x>0$ that
\be N(x)=(\Phi(r+q))^2\int_x^\infty\et^{-\Phi(r+q)y}W^{(r)}(y)\diff y-\Phi(r+q)\et^{-\Phi(r+q)x}W^{(r)}(x)-\et^{-\Phi(r+q)x}W^{(r)\prime}(x).\nn\ee
By \eqref{eq:scale}, \eqref{limbigx} and \eqref{eq:W'} We notice that
\be\lim_{x\to\infty}N(x)=0\\-\lim_{x\to\infty}(\et^{-\Phi(r+q)x}W^{(r)}(x))\bigg(\Phi(r+q)+\lim_{x\to\infty}\frac{W^{(r)\prime}(x)}{W^{(r)}(x)}\bigg)=0.
\ee
On the other hand, for any $x>0$
\[N'(x)=-\et^{-\Phi(r+q)x}W^{(r)\prime\prime}(x), \quad\forall x>0.\]
By  \cite[Theorem 3.4]{Kuznetsov_2011}, we know that when the tail jump measure $\Pi(-\infty,-x)$ has a completely monotone density over $(0,\infty)$, 
there is a constant $a^*$ such that 
$W^{(r)\prime\prime}(x)=-\et^{\Phi(r+q)x}N'(x)<0$ for
$x\in(0,a^*)$ and $W^{(r)\prime\prime}=-\et^{\Phi(r+q)x}N'(x)>0$ for $x\in(a^*,\infty)$. It follows that 
$N(\cdot)$ is strictly increasing over $(0,a^*)$, and is strictly decreasing over $(a^*,\infty)$. Given that $N(x)\to0$ as $x\to\infty$, we know that $N(x)>0$ for at least all $x\in[a^*,\infty)$. Using the monotonicity property of $N(\cdot)$, we know that, 
\begin{enumerate}
\item[(i)]
either $N(x)>0$ provided that $N(0+)\ge0$, in which case $h(\cdot)$ is strictly increasing over $(0,\infty)$;
\item[(ii)] or there is a $x_0\in(0,a^*)$ such that $N(x)<0$ for all $x\in(0,x_0)$ and $N(x)>0$ for all $x\in(x_0,\infty)$. In this case, because $h(0)=-\frac{1}{\Lambda(0)}<0$, $h'(x)<0$ for $x\in(0,x_0)$, and $h'(x)>0$ for $x\in(x_0,\infty)$, we know that there is a unique root to equation $h(x)=0$ over $(0,\infty)$, denoted by $x^\star$, and it holds that $x^\star>x_0$. Hence, we also have $h(\cdot)$ is strictly increasing over $(x^\star,\infty)$.
\end{enumerate}
 So in both cases, Assumption \ref{assumption''}(ii) holds with a finite $x^\star$. Moreover, we compute $$c_0=\lim_{z\to\infty}z\ex[\et^{-rT_z^+-q\int_0^{T_z^+}\ind_{(-\infty,0)}(X_s)\diff s}\ind_{\{T_z^+<\infty\}}]=0.$$
By Remark \ref{rem32} and Theorem \ref{thmSNLP}, we obtain the following result:
\begin{figure}
\centering
\includegraphics[width=3.5in]{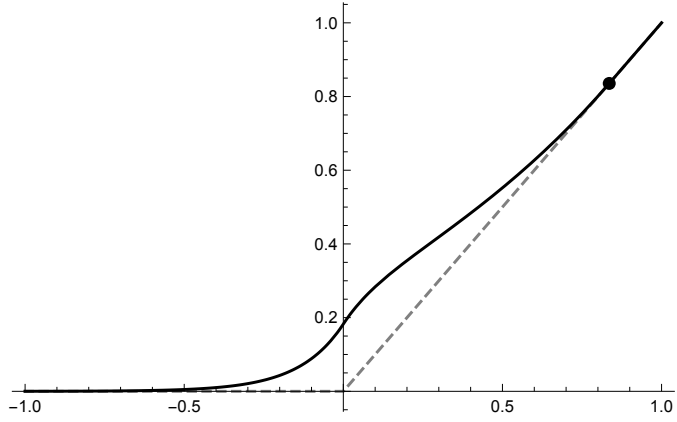}
\caption{Here we plot the reward function $f(x)=(x\vee0)$ (in gray dashed line) and the value function given in \eqref{eq:prop33} (in black solid line).  Laplace exponent used: $\psi(\lambda)=0.18\lambda+0.02\lambda^2-0.25(\frac{\lambda}{\lambda+4})$. Parameter: $r=0.18, q=2$. The optimal threshold $x^\star=0.8356$. According to Theorem \ref{thmSNLP}, smooth fit holds at $x^\star$. The black dot stands for the pasting point where optimal exercising occurs. }
\label{fig2}
\end{figure}
\begin{prop}There is a unique root to equation $x-\frac{1}{\Lambda(x)}=0$ over $(0,\infty)$, which we denote by $x^\star$.  The optimal stopping time for problem \eqref{problem:occupation} is the up-crossing strategy $T_{x^\star}^+$. Moreover, (see Figure \ref{fig2} for a visualization)
\begin{align}
V(x)=&\ex_{x}[e^{-rT_{x^\star}^+-q\int_{0}^{T_{x^\star}^+}\ind_{(-\infty,0)}(X_{s})ds}X_{T_{x^\star}^+}\ind_{\{T_{x^\star}^+<\infty\}}]\nn\\
=&\begin{dcases} x^\star\frac{\int_0^{\infty} \et^{-\Phi(r+q)y}W^{(r)}(x+y)\diff y}{\int_0^{\infty} \et^{-\Phi(r+q)y}W^{(r)}(x^\star+y)\diff y}, &\text{ if $x\le x^\star$},\\
x,&\text{ if $x> x^\star$}.\label{eq:prop33}
\end{dcases}
\end{align}
\end{prop}

\section{Recursive optimal stopping problems}\label{sec:mul}
In optimal stopping, a natural question to ask is whether the optimality of a threshold type strategy for a single optimal stopping problem  will imply that for the corresponding multiple stop version or recursive stop version of the problem.  As seen in \cite{Carmona2008,TimKazuHZ14}, this question is not trivial even for the most popular reward function $(\et^x-K)^+$.  It is thus our objective in this section to provide a positive answer to this question.

\subsection{Cases with constant discount rate, refraction times, and no running cost }\label{sec41}

In this section, we will discuss the optimality of a threshold type strategy in recursive optimal stopping problems with refraction times but no running cost. More specifically, for a constant discounting rate\footnote{\label{fn:CAFr}Our approach does not apply to a general CAF discounting here, because of the ``delay'' caused by the refraction time $\delta$ and  possible spatially inhomogeneity of the law of the associated  random time $\zeta$ in  $X_0$. See the proof of Proposition \ref{prop:induction} for more details. } $r>0$, let $X_\cdot$ be a general L\'evy process henceforth, and let $f(\cdot)$ be a lower semi-continuous function,  satisfying Condition (M) in Definition \ref{conM0}. For
 a positive integer $n\ge 1$,  we consider the following recursive optimal stopping problems
\begin{align}\label{multiplepro}
v^{(l)}(x)&=\sup_{\tau\in\mathcal{T}}\ex_{x}[\et^{-r\tau}f^{(l)}(X_{\tau})\ind_{\{\tau<\infty\}}]\text{, $l=1,2,\ldots,n$},\\
\text{with }f^{(l)}(x)&=f(x)+\ex_x[\et^{-r\delta}v^{(l-1)}(X_{\delta})],\label{multiplepro1}
\end{align}
where $v^{(0)}(x)\equiv 0$, and $\delta>0$ is the refraction time.

\begin{remark}
We point out, in order for the problems in \eqref{multiplepro} to admit finite value functions $v^{(l)}(\cdot)$, $l=1,\ldots,n$, we only need $v^{(1)}(\cdot)\equiv v(\cdot)$ to be well defined and finite.  This is due to the supermartingale property of value function $v(\cdot)$:
\begin{multline*}f(x)\le f^{(2)}(x)\le f(x)+v(x)
\Rightarrow v^{(2)}(x)\le v(x)+\sup_{\tau\in\mc{T}}\ex_x[\et^{-r\tau}v(X_\tau)\ind_{\{\tau<\infty\}}]\le v(x)+v(x)=2v(x),\end{multline*}
\vspace*{-2ex}
\begin{multline*}
f(x)\le f^{(3)}(x)\le f(x)+v^{(2)}(x)\\\Rightarrow v^{(3)}(x)\le v(x)+\sup_{\tau\in\mc{T}}\ex_x[\et^{-r\tau}v^{(2)}(X_\tau)\ind_{\{\tau<\infty\}}]\le v(x)+2v(x)=3v(x).\end{multline*}
\end{remark}

\begin{remark}
The recursive optimal stopping problems defined in \eqref{multiplepro} and \eqref{multiplepro1} naturally arise in optimal stopping problems with multiple exercising opportunities. In particular, consider
\be\label{mulpro1}
V^{(n)}(x):=\sup_{\vec{\tau}\in \mathcal{T}^{(n)}_\delta}\ex_{x}[\sum_{i=1}^n \et^{-r\tau_{i}}f(X_{\tau_i})],\text{ $\forall x\in\R$},
\ee
where $\mathcal{T}^{(n)}_\delta$ is the set of admissible sequence of exercise times, defined as
\begin{align*}
\mathcal{T}^{(n)}_\delta:=\{&\vec{\tau}=(\tau_{n},...,\tau_{1})\in\mathcal{T}^n:\tau_{i+1}+\delta\leq\tau_{i}, \forall i=n-1,\ldots,1\}
\end{align*}
with a constant $\delta>0$ representing the refraction time that separates consecutive exercises. Then it is a standard argument to show that the value function $v^{(n)}(\cdot)$ determined in \eqref{multiplepro} is identical to $V^{(n)}(\cdot)$ in \eqref{mulpro1}, subject to some (extra) mild condition on $f(\cdot)$ (see, e.g., \cite{Carmona2008,TimKazuHZ14}).
\end{remark}

 Obviously, by Theorem \ref{singlesolution}, we know that lower semi-continuity and Condition (M) imply the optimality of the up-crossing strategy $T_{x^\star}^+$ for the optimal stopping problem with reward function $f(\cdot)$.  If one can show that both the lower semi-continuity property and the set $\Upsilon_{\e_r}$ are invariant under the recursive operation as prescribed in \eqref{multiplepro}-\eqref{multiplepro1}, then the optimality of threshold type strategies for the recursive optimal stopping problems \eqref{multiplepro}-\eqref{multiplepro1} can be proved by applying Theorem \ref{singlesolution} sequentially.

\begin{prop}\label{prop:induction}
Suppose the reward function $f(\cdot)$ is lower semi-continuous and belongs to $\Upsilon_{\e_r}$ (defined in Definition \ref{conM0}), then so does $f^{(l)}(\cdot)$ for all $l=1,2,\ldots$. That is, for every positive integer $l$, $f^{(l)}(\cdot)$ is lower semi-continuous, and there is a nondecreasing function $h^{(l)}$ such that  $h^{(l)}(x)>0$ if and only if $x>x_l^\star$ for some constant $x_l^\star\in[-\infty,\infty)$, and
 \be
 \ex_x[|h^{(l)}(\ol{X}_{\e_r})|]<\infty, \quad f^{(l)}(x)=\ex_x[h^{(l)}(\ol{X}_{\e_r})], \quad\forall x\in\mathbb{R}. \nn
 \ee
 Moreover,  a recursion equation hold: for any $l=1,2,\ldots$,
\be
h^{(l+1)}(x)=h(x)+\et^{-r\delta}\ex[h^{(l)}(x+X_\delta)\ind_{\{x+X_\delta>x_l^\star\}}],\quad\forall x\in\R.\nn\ee
\end{prop}
\begin{proof}
We apply mathematical induction to prove the claim. To that end, we assume that, for some positive integer $k$, $f^{(l)}(\cdot)\in\Upsilon_{\e_r}$ holds for all $l=1,2,\ldots,k$. By Theorem \ref{singlesolution}, we know that the value function for the reward  function $f^{(k)}(\cdot)$ is given by
\be
v^{(k)}(x)=\ex_x[h^{(k)}(\ol{X}_{\e_r})\ind_{\{\ol{X}_{\e_r}>x_k^\star\}}]=\ex[h^{(k)}(x+\ol{X}_{\e_r})\ind_{\{x+\ol{X}_{\e_r}>x_k^\star\}}],\quad\forall x\in\mathbb{R}.\nn
\ee
Let us denote by $Y$ a random variable that is independent of $X_\cdot$, having the same law as $\ol{X}_{\e_r}$ under $\pr$. Then we have
\[v^{(k)}(x)=\ex[h^{(k)}(x+Y)\ind_{\{x+Y>x_k^\star\}}],\quad\forall x\in\R.\]
Hence, by \eqref{multiplepro1} we have
\begin{align}
f^{(k+1)}(x)&=f(x)+\et^{-r\delta}\ex[v^{(k)}(x+X_{\delta})]\nn\\
&=\ex[h^{(1)}(x+\ol{X}_{\e_r})]+\ex[\ex[\et^{-r\delta}h^{(k)}(x+X_\delta+Y)\ind_{\{x+X_\delta+Y>x^\star_k\}}|X_\delta]]\nn\\
&=\ex[h^{(1)}(x+Y)]+\ex[\et^{-r\delta}h^{(k)}(x+X_\delta+Y)\ind_{\{x+X_\delta+Y>x_k^\star\}}]\nn\\
&=\ex[h^{(k+1)}(x+Y)]=\ex_x[h^{(k+1)}(\ol{X}_{\e_r})],\label{eq:30}
\end{align}
where $h^{(k+1)}$ is defined as in the statement of the proposition. 
Because both $h^{(1)}(\cdot)$ and $h^{(k)}(\cdot)$ are nondecreasing, we know from the definition of $h^{(k+1)}(\cdot)$ that it is also nondecreasing. Moreover, from \eqref{eq:30} we also see the integrability condition of $h^{(k+1)}(\cdot)$ holds.

On the other hand, for any $x>x_1^\star$, we know from $h^{(k)}(\cdot)\ind_{\{\cdot>x_k^\star\}}>0$ that
\be h^{(k+1)}(x)\ge h^{(1)}(x)>0.\label{ineq_h}\ee
This implies that there exists an $x_{k+1}^\star\in [-\infty, x_1^\star]$ such that, $h^{(k+1)}(x)>0$ if and only if $x>x_{k+1}^\star$.

Finally, we need to prove that $f^{(k+1)}(\cdot)$ is lower semi-continuous, or equivalently, left continuous since it is nondecreasing. Using Lemma \ref{lem:cont} we know that $v^{(k)}(\cdot)$ is nondecreasing and lower semi-continuous. Therefore, for any $x_1>x_2$ we have $\pr$-a.s that, 
$$
v^{(k)}(x_1+X_\delta)\ge \limsup_{x_2\uparrow x_1}v^{(k)}(x_2+X_\delta)\ge \liminf_{x_2\uparrow x_1}v^{(k)}(x_2+X_\delta)=v^{(k)}(x_1+X_\delta),$$
which implies that $v^{(k)}(x_1+X_\delta)= \lim_{x_2\uparrow x_1}v^{(k)}(x_2+X_\delta)$ holds $\pr$-a.s.

Moreover, by the supermartingale  property of value function $v^{(k)}(\cdot)$, we know that nonnegative random variable $v^{(k)}(x_1+X_\delta)$ has a finite expectation.
By the dominated convergence theorem, we have
\begin{align}
\lim_{x_2\uparrow x_1}\ex[v^{(k)}(x_2+X_\delta)]=\ex[\lim_{x_2\uparrow x_1}v^{(k)}(x_2+X_\delta)]=\ex[v^{(k)}(x_1+X_\delta)].\nn
\end{align}
Thus,  $\ex[v^{(k)}(x+X_\delta)]$ is left continuous in $x$. The left continuity of $f^{(k+1)}(\cdot)$ now follows from the first line of \eqref{eq:30}.
This completes the proof.
\end{proof}
%

\begin{prop}\label{prop:42}
The sequence $\{x_{l}^\star\}_{1\le l\le n}$ in Proposition \ref{prop:induction} is non-increasing in $l$. The sequence of functions $\{h^{(l)}(\cdot)\}_{1\le l\le n}$ in Proposition \ref{prop:induction} is nondecreasing in $l$.
\end{prop}
\begin{proof}
First, from the proof of Proposition \ref{prop:induction}, we already know the claim holds for $l=1,2$. Suppose it holds for $l=1,2,\ldots,k$ for some $k\ge2$. In particular,
\[h^{(k)}(x)\ge h^{(k-1)}(x),\quad\forall x\in\mathbb{R}, \text{ and }x_k^\star\le x_{k-1}^\star.\]
Then, from
\begin{align}\nonumber
h^{(k+1)}(x)=&h^{(1)}(x)+\ex[\et^{-r\delta}h^{(k)}(x+X_\delta)\ind_{\{x+X_\delta>x_k^\star\}}],
\nn\\
h^{(k)}(x)=&h^{(1)}(x)+\ex[\et^{-r\delta}h^{(k-1)}(x+X_\delta)\ind_{\{x+X_\delta>x_{k-1}^\star\}}],\nn
\end{align}
we know that
\begin{align}
&h^{(k+1)}(x)-h^{(k)}(x)\nn\\
=&\et^{-r\delta}\ex[h^{(k)}(x+X_{\delta})\ind_{\{x+X_{\delta}>x_{k}^*\}}-h^{(k-1)}(x+X_{\delta})\ind_{\{x+X_{\delta}>x^{\star}_{k-1}\}}]\nn\\
\nn
=&\et^{-r\delta}\left(\ex[(h^{(k)}(x+X_{\delta})-h^{(k-1)}(x+X_{\delta}))\ind_{\{x+X_{\delta}>x_{l-1}^*\}}]+\ex[h^{(k)}(x+X_\delta)\ind_{\{x_{k-1}^\star\ge x+X_\delta>x_k^\star\}}]\right)\nn\\
\ge&\et^{-r\delta}\ex[(h^{(k)}(x+X_{\delta})-h^{(k-1)}(x+X_{\delta}))\ind_{\{x+X_{\delta}>x_{l-1}^*\}}]\ge0.\nn
\end{align}
Therefore, we have
\be\label{h_+1>h}
x_{k+1}^{\star}\equiv\sup\{x\in\mathbb{R}: h^{(k+1)}(x)\le 0\}\le \sup\{x\in\mathbb{R}: h^{(k)}(x)\le 0\}\equiv x_k^\star.
\ee
This completes the proof.
\end{proof}

\begin{remark}
There are easily verifiable sufficient conditions that lead to a  strictly decreasing sequence of thresholds
\[x_1^\star>x_2^\star>\ldots.\]
For example, if $X_\cdot$ has unbounded variation and $h^{(1)}$ is continuous on $\mathbb{R}$, then we are in this case. To see this, we apply mathematical induction to show that  $h^{(l)}$ is continuous over $\mathbb{R}$ for all $l\ge1$. On the other hand, by
  \cite[Theorem 24.10(i)]{Sato99}, we know that the random variable $X_\delta$ is supported on $\R$.
Hence, the first inequality \eqref{ineq_h} becomes a strict inequality and in particular, $h^{(2)}(x)>h^{(1)}(x)$ for all $x\in\mathbb{R}$. As a consequence, we know that
$x_2^\star<x_1^\star$. By the same argument, we see that the sequence $\{x_l^\star\}_{l\ge1}$ is strictly decreasing in $l$.
\end{remark}

In conclusion, we obtain the following result.
\begin{thm}\label{thm44}
Assume that the reward function $f(\cdot)=f^{(1)}(\cdot)$ is lower semi-continuous and satisfies Condition (M). Then the recursive optimal stopping problems \eqref{multiplepro} and \eqref{multiplepro1} are solved by up-crossing strategies $T_{x_l^\star}^+$. And these thresholds satisfy
\[-\infty\le x_n^\star\le x_{n-1}^\star\le \ldots\le x_1^\star<\infty.\]
\end{thm}

Finally, we give an example of swing options under general L\'evy processes.
\begin{cor}\label{example42}
Suppose the discount factor rate is a constant $r>0$ such that $\ex[\et^{X_1}]<\et^r$. For any $K_1, K_2>0$, the results in Theorem \ref{thm44} hold if the reward function  $f^{(1)}(\cdot)$  can be written as a convex combination of $\et^x-K_1$ and $K_2-\et^{-x}$.
\end{cor}
\begin{proof} It suffices to prove that both $f_1(x)=\et^x-K_1$ and $f_2(x)=K_2-\et^{-x}$ satisfy Condition (M).  The claim for the convex combination $\alpha f_1(x)+\beta f_2(x)$ follows from Remark \ref{rem:convex}.

From  \cite[Theorem 1]{mordecki2002} we know that $f_1(x)=\ex_x[h_1(\ol{X}_{\e_r})]$ for $h_1(x)=\et^x/\ex[\et^{\ol{X}_{\e_r}}]-K_1$, and $h_1(x)\gtrless 0$ if and only if $x\gtrless  \log (K_1\ex[\et^{\ol{X}_{\e_r}}])$.

Similarly, from  \cite[Theorem 2]{mordecki2002} we know that $\ex[\et^{-\ol{X}_{\e_r}}]=\ex[\et^{\inf_{s\in[0,\e_r]}(-X_s)}]>0$. Then it is straightforward to verify that
 $f_2(x)=\ex_x[h_2(\ol{X}_{\e_r})]$ for $h_2(x)=K_2-\et^{-x}/\ex[\et^{-\ol{X}_{\e_r}}]$, and $h_2(x)\gtrless 0$ if and only if $x\gtrless \log (K_2\ex[\et^{-\ol{X}_{\e_r}}])$.

\end{proof}
\subsection{Cases with random discounting and a running cost}\label{sec42}
In a recent work \cite{YamazakiAMO2014}, the author considered the following type of multiple stopping problem:
\be
\sup_{\substack{\tau^{(1)}\le\ldots\le\tau^{(n)}\\
\tau^{(l)}\in\mc{T}}} \sum_{l=1}^{n}\ex_x[\int_{\tau^{(l-1)}}^{\tau^{(l)}}\et^{-rt}(-C_l(X_t))\diff t+\et^{-r\tau^{(l)}}f_l(X_{\tau^{(l)}})\ind_{\{\tau^{(l)}<\infty\}}],\nn
\ee
where we interpret $C_l(\cdot)$ as the running cost between the $(l-1)$-th and the $l$-th stoppings, and $f_l(\cdot)$ as the reward upon the $l$-th stopping, and $r>0$ as the discounting rate. In particular, it is shown using explicit calculations (and the principle of smooth fit) under the spectrally negative L\'evy model that the optimal strategies $\tau^{(l)}$ are of threshold type.

In this section, we let $X_\cdot$ be a general L\'evy process, and  use the average problem approach to show that threshold type strategy is optimal for a related recursive optimal stopping problems. Formally, consider
\begin{align}\label{multiplepro11}
v^{(l)}(x)&=\sup_{\tau\in\mathcal{T}}\ex_{x}[\int_0^{\tau}\et^{-A_t}(-C_l(X_t))\diff t+\et^{-A_\tau}f^{(l)}(X_{\tau})\ind_{\{\tau<\infty\}}]\text{, $l=1,2,\ldots,n$.}\\
\text{with }f^{(l)}(x)&=f_l(x)+v^{(l-1)}(x),\label{multiplepro12}
\end{align}
where $C^{(0)}(x)=v^{(0)}(x)\equiv0$ as before, and $A_\cdot$ is the CAF random discounting.

\begin{Assumption}\label{assumption4}
For all $1\le l\le n$, the cost functions $C_l(x)$ satisfy
\[\ex_x[\int_0^\infty \et^{-A_t}|C_l(X_t)|\diff t]<\infty.\]
\end{Assumption}

\begin{remark}\label{rmk:AC}
The condition in Assumption \ref{assumption4} holds when, for example, $C_l(\cdot)$ is uniformly bounded and $\ex_x[\et^{-A_t}]$ is integrable over $t\in(0,\infty)$. If $X_\cdot$ is spectrally negative and it holds that $A_t\ge rt$, for all $t\ge0$ holds $\pr_x$-a.s. for some $r>0$, then by \cite[Corollary 8.9]{Kyprianou2006} we have
\[0\le \ex_x[\int_0^\infty \et^{-A_t}|C_l(X_t)|\diff t]\le \ex_x[\int_0^\infty \et^{-rt}|C_l(X_t)|\diff t]=\int_{-\infty}^\infty\left(\Phi'(r)\et^{-\Phi(r)(y-x)}-W^{(r)}(x-y)\right)|C_l(y)|\diff y,\]
where $W^{(r)}(\cdot)$ is the $r$-scale function of $X_\cdot$ and $\Phi'(r)$ is the derivative of $\Phi(r)$ (see Appendix \ref{sec:SNLP}).  Thus, Assumption \ref{assumption4} holds if $\int_{-\infty}^\infty (\Phi'(r)\et^{-\Phi(r)(y-x)}-W^{(r)}(x-y))|C_l(y)|\diff y<\infty$. More sufficient conditions can be found   in \cite{YamazakiAMO2014}.
\end{remark}
To apply the average problem approach, we need to recast the problem \eqref{multiplepro11}-\eqref{multiplepro12} to one without the running cost $-C_l(\cdot)$. To that end, we define for any $1\le l\le n$ that
\be
\ol{C}_l(x):=\ex_x[\int_0^\infty\et^{-A_t}C_l(X_t)\diff t].\nn
\ee
We also define that $\ol{C}_0(x)\equiv0$.
It follows that, for any stopping time $\tau\in\mc{T}$, and $1\le l\le n$ such that $f^{(l)}(\cdot)$ is well defined using \eqref{multiplepro12}, we have
\begin{align}
&\ex_x[\int_0^\tau\et^{-A_t}(-C_l(X_t))\diff t+\et^{-A_\tau}f^{(l)}(X_\tau)\ind_{\{\tau<\infty\}}]\nn\\
=&\ex_x[\int_0^\infty\et^{-A_t}(-C_l(X_t))\diff t-\ind_{\{\tau<\infty\}}\int_{\tau}^\infty\et^{-A_t}(-C_l(X_t))\diff t+\et^{-A_\tau}f^{(l)}(X_\tau)\ind_{\{\tau<\infty\}}]\nn\\
=&\ex_x[\int_0^\infty\et^{-A_t}(-C_l(X_t))\diff t-\et^{-A_\tau}\ex_{X_\tau}[\int_{0}^\infty\et^{-A_t}(-C_l(X_t))\diff t]\ind_{\{\tau<\infty\}}+\et^{-A_\tau}f^{(l)}(X_\tau)\ind_{\{\tau<\infty\}}]\nn\\
=&-\ol{C}_l(x)+\ex_x[\et^{-A_\tau}g^{(l)}(X_\tau)\ind_{\{\tau<\infty\}}],\nn
\end{align}
where we used the strong Markov property of $X_\cdot$ in the second equality, and defined
\be g^{(l)}(x):=\ol{C}_l(x)+f^{(l)}(x).\label{eq:CG}\ee
It follows that (using \eqref{multiplepro12} and \eqref{eq:CG})
\begin{align}
v^{(l)}(x)=&-\ol{C}_l(x)+\sup_{\tau\in\mc{T}}\ex_x[\et^{-A_\tau}g^{(l)}(X_\tau)\ind_{\{\tau<\infty\}}],\label{eq:33}\\
g^{(l+1)}(x)=&\ol{C}_{l+1}(x)-\ol{C}_l(x)+f_{l+1}(x)+\sup_{\tau\in\mc{T}}\ex_x[\et^{-A_\tau}g^{(l)}(X_\tau)\ind_{\{\tau<\infty\}}],\label{eq:34}
\end{align}
whenever the right hand sides are well defined. 
By \eqref{eq:33} and \eqref{eq:34} we obtain the main results of this section.
\begin{prop}\label{prop43}
Under Assumption \ref{assumption4}, and assume that $\ol{C}_{l}(\cdot)-\ol{C}_{l-1}(\cdot)+ f_l(\cdot)$ is continuous and belongs to $\Upsilon_{\zeta}$ for all $1\le l\le n$. Then for $1\le l\le n$, functions $g^{(l)}(\cdot)$ are well defined and belong to $\Upsilon_{\zeta}$; moreover, $g^{(l)}(\cdot)$ satisfies 
\be
\lim_{n\to\infty}g^{(l)}(X_{\tau_n})=g^{(l)}(X_\tau)\quad\text{on }\{\tau<\infty\},\label{eq:41}
\ee
for any stopping time $\tau\in\mathcal{T}$ and  a nonincreasing sequence of stopping times $(\tau_n)_{n\ge1}$ such that $\tau_n\downarrow\tau$ as $n\to\infty$.
 And there are thresholds $\{x_l^\star\}_{1\le l\le n}$ such that $v^{(l)}(x)+\ol{C}_l(x)=g^{(l)}(x)$ if and only if $x\ge x_l^\star$. That is, the threshold type strategy $T_{x_l^\star}^+$ is optimal for problems \eqref{multiplepro11} and \eqref{multiplepro12}.
\end{prop}
\begin{proof}
In light of \eqref{eq:CG} and Theorem \ref{singlesolution}, the claim for $l=1$ is obvious. In particular, the condition in \eqref{eq:41} follows from lower semi-continuity of $g^{(1)}(x)=\ol{C}_1(x)+f_1(x)$. 

Suppose the claim holds for $l=k$ for some $k\ge1$, then we know that there is $h_g^{(k)}(\cdot)$ satisfying conditions in Definition \ref{conM0}, and
\[g^{(k)}(x)=\ex_x[h_g^{(k)}(\ol{X}_{\zeta})],\quad\forall x\in\R,\]
and $h_g^{(k)}(x)>0$ if and only if $x>x_k^\star$ for some $x_k^\star\in[-\infty,\infty)$. 
To obtain a similar representation for $\sup_{\tau\in\mc{T}}\ex_x[\et^{-A_\tau}g^{(k)}(X_\tau)\ind_{\{\tau<\infty\}}]$, we cannot apply Theorem \ref{singlesolution} directly (even though we already know that $T_{x_k^\star}^+$ is optimal), as we don't know if $g^{(k)}(\cdot)$ is lower semi-continuous (except when $k=1$). However, as discussed in Remark \ref{rmk:33}, the property in \eqref{eq:41} of $g^{(k)}(\cdot)$ ensures that the proof of Theorem \ref{singlesolution} goes through, so we obtain 
\[\ex_x[\et^{-A_\tau}g^{(k)}(X_\tau)\ind_{\{\tau<\infty\}}]=\ex_x[h_g^{(k)}(\ol{X}_\zeta)\ind_{\{\ol{X}_\zeta>x_k^\star\}}].\]
On the other hand, by the assumption,  there  is $h_f^{(k+1)}(\cdot)$ satisfying conditions in Definition \ref{conM0}, and
\be\ol{C}_{k+1}(x)-\ol{C}_k(x)+f_{k+1}(x)=\ex_x[h_f^{(k+1)}(\ol{X}_\zeta)],\quad\forall x\in\R,\label{eq:hfdef}\ee
and $h_f^{(k+1)}(x)>0$ if and only if $x>\ol{x}_{k+1}^\star$ for some $\ol{x}_{k+1}^\star\in[-\infty,\infty)$.
Now from  \eqref{eq:34}, we have
\begin{align}
g^{(k+1)}(x)=&\ol{C}_{k+1}(x)-\ol{C}_k(x)+f_{k+1}(x)+\ex_x[h_g^{(k)}(\ol{X}_{\zeta})\ind_{\{\ol{X}_{\zeta}>x_k^\star\}}]\nn\\
=&\ex_x[h_f^{(k+1)}(\ol{X}_{\zeta})+h_g^{(k)}(\ol{X}_{\zeta})\ind_{\{\ol{X}_{\zeta}>x_k^\star\}}]\nn\\
=&\ex_x[h_g^{(k+1)}(\ol{X}_{\zeta})],\label{eq:hg1}
\end{align}
where
\be
h_g^{(k+1)}(x):=h_f^{(k+1)}(x)+h_g^{(k)}(x)\ind_{\{x>x_k^\star\}}.\label{eq:hg2}\ee
By Remark \ref{rem:convex}, we know that $g^{(k+1)}(\cdot)\in\Upsilon_\zeta$. In particular,  there is a unique $x_{k+1}^\star\in\R$ such that $h_g^{(k+1)}(x)>0$ if and only if $x>x_{k+1}^\star$. 

We now prove that $g^{(k+1)}(\cdot)$ satisfies the property in \eqref{eq:41}. To that end, denote
\[\widehat{g^{(k)}}(x):=v^{(k)}(x)+\ol{C}_k(x)=\sup_{\tau\in\mc{T}}\ex_x[\et^{-A_\tau}g^{(k)}(X_\tau)\ind_{\{\tau<\infty\}}].\]
We already know that $(\et^{-A_t}\widehat{g^{(k)}}(X_t))_{t\ge0}$ is the Snell envelope of right continuous process $(\et^{-A_t}g^{(k)}(X_t))_{t\ge0}$ (due to \eqref{eq:41}), so it is c\`adl\`ag (see \cite[page 353]{KaratzasShreve01}). In particular, using Lemma \ref{lem:inf} we have  on $\{\tau<\infty\}$ that,
\begin{align}\lim_{n\to\infty}\widehat{g^{(k)}}(X_{\tau_n})&=\exp(A_\tau)\cdot\lim_{n\to\infty}\exp(-A_{\tau_n})\cdot\lim_{n\to\infty}\widehat{g^{(k)}}(X_{\tau_n})=\exp(A_\tau)\cdot\lim_{n\to\infty}\exp(-A_{\tau_n})\widehat{g^{(k)}}(X_{\tau_n})\nn\\
&=\exp(A_\tau)\cdot\exp(-A_\tau)\widehat{g^{(k)}}(X_\tau)=\widehat{g^{(k)}}(X_\tau),\label{eq:44}\end{align}
where the third step is due to the  c\`adl\`ag property.  It follows that, on $\{\tau<\infty\}$,  
\begin{align}
\lim_{n\to\infty}g^{(k+1)}(X_{\tau_n})&=\lim_{n\to\infty}\left(\ol{C}_{k+1}(X_{\tau_n})-\ol{C}_k(X_{\tau_n})+f_{k+1}(X_{\tau_n})\right)+\lim_{n\to\infty}\widehat{g^{(k)}}(X_{\tau_n})\nn\\
&=\left(\ol{C}_{k+1}(X_{\tau})-\ol{C}_k(X_{\tau})+f_{k+1}(X_{\tau})\right)+\widehat{g^{(k)}}(X_{\tau})=g^{(k+1)}(X_\tau).
\end{align}
By mathematical induction, the claim holds for all $l=1,2,\ldots,n.$
\end{proof}

\begin{figure}
\centering
\subfloat[$C(x)\equiv-0.02$]{{\includegraphics[width=200pt]{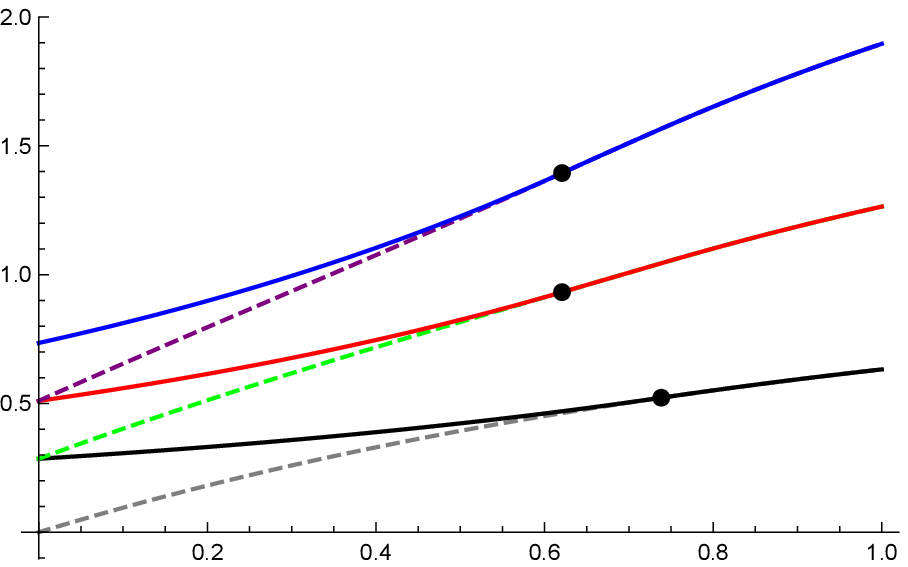} }}%
    \qquad
    \subfloat[$C(x)\equiv0.02$]{{\includegraphics[width=200pt]{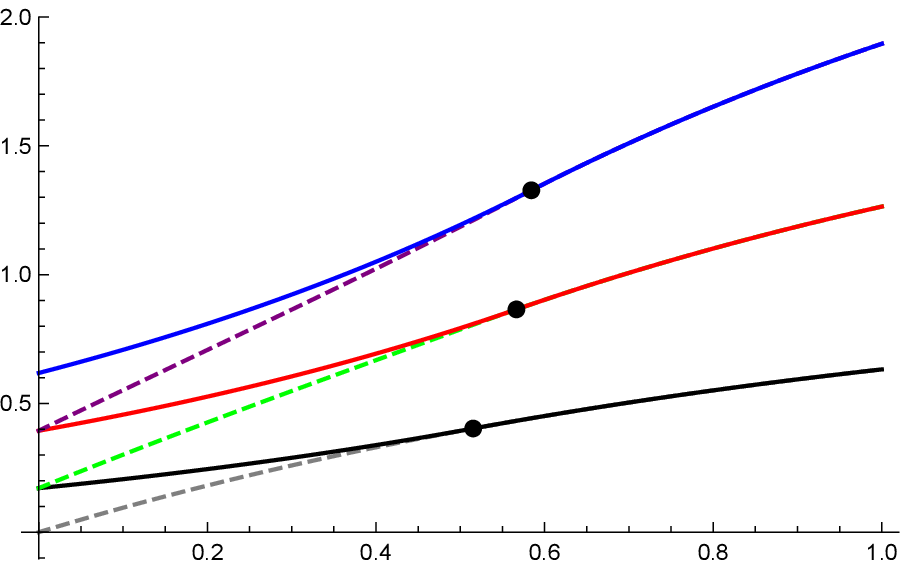} }}%
    \caption{Here we plot the reward function $f^{(l)}(\cdot)$ (in purple, green and gray dashed lines) and the value function $v^{(l)}(\cdot)$ given in \eqref{multiplepro11} (in blue, red, black solid lines), for $l=3,2,1$, respectively.  L\'evy process used: a spectrally negative L\'evy process with Laplace exponent: $\psi(\lambda)=0.18\lambda+0.02\lambda^2-0.25(\frac{\lambda}{\lambda+4})$. Parameter: constant discounting rate $r=0.18$ (hence $A_t=0.18t$ for all $t\ge0$), reward functions $f_1(x)\equiv f_2(x)\equiv f_3(x)\equiv f(x)=1-\et^{-x}$ and the common cost function $C_1(x)\equiv C_2(x)\equiv C_3(x)\equiv C(x)$ with $C(x)\equiv-0.02$ (left) or $C(x)\equiv0.02$ (right). By using Theorem \ref{thmSNLP} with mathematical induction, we know that smooth fit holds at the pasting points. The black dots stand for the pasting points where optimal exercising occurs. In the left figure, $x^\star=0.6207=x_3^\star=x_2^\star<x_1^\star=0.7384$; in the right figure,  $x_1^\star=0.5153$, $x_2^\star=0.5666$, $x_3^\star=0.5843$ and $x^\star=0.6207$. }%
    \label{fig:example}%
\end{figure}

\begin{remark}\label{rmk:special}
A special case that can be conveniently checked against conditions in Proposition \ref{prop43} is the situation when we have identical cost functions and identical reward functions, i.e. \be
C_l(\cdot)\equiv C_1(\cdot), f_l(\cdot)\equiv f_1(\cdot), \forall l\ge 2,\label{eq:special_eq}\ee
where it is assumed that $\ol{C}_1(\cdot)+f_1(\cdot), f_1(\cdot)$ are continuous and belong to $\Upsilon_{\zeta}$. In the case of spectrally negative L\'evy with constant discounting rate $r>0$, we can construct a class of such examples with $C_1(x)=c\cdot e^{\beta x}$ for $c>0$ and $0\le\beta<\Phi(r)$, and $f_1(x)\in\Upsilon_{\mathbf{e}_r}$.
\end{remark}
\begin{prop}
\label{prop44}
Under the condition of Proposition \ref{prop43}, for $1\le l\le n$, let $\ol{x}_{l}^\star$ be the unique constant in $[-\infty,\infty)$ such that $$h_f^{(l)}(x)\begin{dcases}>0,\quad & x>\ol{x}_l^\star,\\\le0, & x\le \ol{x}_l^\star, \end{dcases}$$ where $h_f^{(l)}(\cdot)$ is given in \eqref{eq:hfdef} when $l=k+1$. Then we have $x_1^\star=\ol{x}_1^\star$, and 
for any $k\ge2$:
\begin{enumerate}
\item[(i)] if  $\ol{x}_{k}^\star< x_{k-1}^\star$, then $x_{k}^\star=\ol{x}_{k}^\star$;
\item[(ii)] if  $\ol{x}_{k}^\star\ge x_{k-1}^\star$, then $x_{k-1}^\star\le x_{k}^\star\le\ol{x}_{k}^\star$.
\end{enumerate}
 In particular, in the special case that $\ol{x}_l^\star\equiv x^\star$ for all $2\le l\le n$:
 \begin{enumerate} 
 \item if $x^\star< x_1^\star$,
 then thresholds $\{x_l^\star\}_{1\le l\le n}$ satisfy
\[x^\star= x_n^\star=\ldots= x_2^\star< x_1^\star;\]
\item if $x^\star\ge x_1^\star$, then thresholds $\{x_l^\star\}_{1\le l\le n}$ satisfy
\[x_1^\star\le x_2^\star\le \ldots\le x_n^\star\le x^\star.\]
See Figure \ref{fig:example} for a visualization.
\end{enumerate}
\end{prop}
\begin{proof}
The claim that $x_1^\star=\ol{x}_1^\star$ obviously holds. 
For $k\ge2$, if $\ol{x}_{k}^\star< x_{k-1}^\star$, then $h_f^{(k)}(x)>0$ and $h_g^{(k-1)}(x)>0$ for all $x>x_{k-1}^\star$,
so by  \eqref{eq:hg2} we know that $h_g^{(k)}(x)>0$ for all $x>x_{k-1}^\star$. This implies that $x_{k}^\star\le x_{k-1}^\star$. On the other hand, for all $x\le x_{k-1}^\star$, we have $h_g^{(k)}(x)=h_f^{(k)}(x)$, so 
 $x_{k}^\star=\ol{x}_{k}^\star$.

 If $\ol{x}_k^\star\ge x_{k-1}^\star$, then by a similar argument, we know that $h_g^{(k)}(x)>0$ for all $x>\ol{x}_{k}^\star$, so $x_{k}^\star\le\ol{x}_{k}^\star$. On the other hand,  for all $x\le x_{k-1}^\star$, $h_g^{(k)}(x)=h_f^{(k)}(x)\le 0$ . So we know that $x_{k-1}^\star\le x_{k}^\star\le \ol{x}_{k}^\star$. 
 
 The results in the special case now follow from mathematical induction.
\end{proof}
\begin{remark}\label{rmk:360}
Suppose Eq. \eqref{eq:special_eq} holds, and $\ol{C}_1(\cdot)+f_1(\cdot), f_1(\cdot)$ are continuous and belong to $\Upsilon_{\zeta}$, then we have $\ol{x}_l^\star\equiv x^\star$ for all $2\le l\le n$, so we are in the special case of Proposition \ref{prop44}. Specifically, we can write 
\[\ol{C}_1(x)+f_1(x)=\ex_x[h_1(\ol{X}_\zeta)], \quad f_1(x)=\ex_x[h_2(\ol{X}_\zeta)], \quad\forall x\in\R,\]
where $h_1(\cdot)$ and $h_2(\cdot)$ are nondecreasing functions prescribed in Definition \ref{conM0}. Moreover, we can define $x_1^\star$ and $x^\star$ with
\[h_1(x)\begin{dcases}
>0, \quad &x>x_1^\star,\\
\le 0,&x\le x_1^\star,
\end{dcases}
\quad
h_2(x)\begin{dcases}
>0, \quad &x>x^\star,\\
\le 0,&x\le x^\star.
\end{dcases}
\]
In this case, Proposition \ref{prop44} gives us a very simple rule for solving the optimal stopping problems \eqref{multiplepro11}-\eqref{multiplepro12}, via the comparison between the optimal threshold for $l=1$, $x_1^\star$, and  $x^\star$: 

If $x^\star<x_1^\star$, then it will be optimal to use the first $(n-1)$ opportunities all at once upon crossing the threshold $x^\star$ from below, and use the last exercising opportunity only when the underlying process crosses $x_1^\star$ from below. The intuition behind this is that the cost component $-C_1(\cdot)$ is actually beneficial, in the sense that the perpetuity $-\ol{C}_1(\cdot)$ tend to be positive (a constant perpetuity function $\ol{C}_1(\cdot)\equiv \ol{C}_1<0$ will cause $x_1^\star$ to be larger than $x^\star$). In this case, one should be slightly patient, and use the last exercising opportunity when the underlying reaches a higher target $x_1^\star$.  See Figure \ref{fig:example}(a).

If $x^\star\ge x_1^\star$, then it will be optimal to exercise all $n$ opportunities upon crossing $x_n^\star$ from below. In this case, the cost component $-C_1(\cdot)$ is not beneficial, in the sense that the perpetuity $-\ol{C}_1(\cdot)$ tend to be non-positive (a constant perpetuity function $\ol{C}_1(\cdot)\equiv \ol{C}_1\ge0$ will  cause $x_1^\star$ to be no larger than $x^\star$), so one should exercise all opportunities as soon as possible.  See Figure \ref{fig:example}(b).

In both cases, it is always optimal to exercise at least the first $(n-1)$ opportunities at the same time. This is because the cost functions are the same, and there is no benefit to make separate, consecutive stops except possibly the last one. 

The same intuition works for the general cases of Proposition \ref{prop44}. For example, $\ol{x}_k^\star<x_{k-1}^\star$ means that the cost to be paid between the $(n+1-k)$-th and the $(n+2-k)$-th exercises is beneficial, suggesting separate, consecutive stops; in contrast, $\ol{x}_k^\star\ge x_{k-1}^\star$ means that there is no benefit to separately use these two exercising opportunities.
\end{remark}
\begin{cor}
Assume that the discounting rate is constant $r>0$, and the reward functions $(f_l(\cdot))_{1\le l\le n}$ are continuous and belong to $\Upsilon_{\e_r}$, and $f_1(\cdot)$ has representation $h(\cdot)$. Suppose that for all $1\le l\le n$, the cost functions $C_l(x)\equiv C(x)=-L+\sum_{i=1}^k c_i\et^{\alpha_i x}$, where $L>0$, $c_i, \alpha_i>0$,  $\ex[\et^{\alpha_iX_1}]<\infty$ and $\psi(\alpha_i)<r$ for all $i=1,2,\ldots, k$ or $C(x)\equiv$some constant in $(-r\cdot h(\infty), \infty)$. Then results in Propositions \ref{prop43} and \ref{prop44} hold.
\end{cor}
\begin{proof} In this case, we notice that $\ol{C}_{l}(\cdot)-\ol{C}_{l-1}(\cdot)+f_l(\cdot)=f_l(\cdot)$ for $2\le l\le n$, so we only need to verify that $\ol{C}(\cdot)+f_1(\cdot)\in\Upsilon_{\e_r}$. If $C(\cdot)$ is a constant function bounded in $(-r\cdot h(\infty), \infty)$, then the conclusion obviously holds. Otherwise, we calculate $\ol{C}(\cdot)$ as follows:
\begin{align}
\ol{C}(x)=\ex_x[\int_0^\infty\et^{-rt}C(X_t)\diff t]=&-\frac{L}{r}+\sum_{i=1}^kc_i\int_0^{\infty}\et^{-rt}\ex_{x}[\et^{\alpha_i X_t}]\diff t\nn\\
=&-\frac{L}{r}+\sum_{i=1}^kc_i\et^{\alpha_ix}\int_0^\infty \et^{-(r-\psi(\alpha_i))t}\diff t\nn\\
=&-\frac{L}{r}+\sum_{i=1}^k\frac{c_i}{r-\psi(\alpha_i)}\et^{\alpha_ix}.
\end{align}
One can also see from the above calculation that Assumption \ref{assumption4} holds.
On the other hand, similar as in the proof of Corollary \ref{example42}, we have
\be
\ol{C}(x)=\ex_x[h_c(\ol{X}_{\e_r})]\quad\text{ where }\quad h_c(x):=\sum_{i=1}^k\frac{c_i}{r-\psi(\alpha_i)}\frac{\et^{\alpha_i x}}{\ex[\et^{\alpha_i\ol{X}_{\e_r}}]}-\frac{L}{r}.\nn
\ee
So the function $\ol{C}(\cdot)+f_1(\cdot)\in\Upsilon_{\e_r}$. Thus, results in Propositions \ref{prop43} and \ref{prop44} apply. \end{proof}

{\bf Acknowledgements} We are grateful to the associate editor and  anonymous referees who offered helpful comments that improve this work.

\appendix
\section{Proofs}\label{sec:proof}
\begin{proof}[Proof of Lemma \ref{lem:equivalence}]
As $z\to\infty$, we know that $T_z^+\to\infty$, $\pr_x$-a.s. Therefore, if Assumption \ref{ass:A} holds, then
\begin{enumerate}
\item if $\pr_x(A_\infty=\infty)=1$, we have
\be
0\le \lim_{z\to\infty}\ex_x[\exp(-A_{T_z^+})\ind_{\{T_z^+<\infty\}}]\le 
\lim_{z\to\infty}\ex_x[\exp(-A_{T_z^+})]= \ex_x[\exp(-A_\infty)]=0,\nn
\ee
where we used the bounded convergence theorem to get the first equality;
\item if $\pr(\limsup_{t\to\infty}X_t<\infty)=1$, we know that $\ol{X}_\infty$ is an almost surely finite random variable, so
\be
0\le \lim_{z\to\infty}\ex_x[\exp(-A_{T_z^+})\ind_{\{T_z^+<\infty\}}]\le 
\lim_{z\to\infty}\pr_x(T_z^+<\infty)=\lim_{z\to\infty} \pr_x(\ol{X}_\infty>z)=0.\nn
\ee
\end{enumerate}
On the other hand, from the additive and nonnegative property of $A_\cdot$, we know that $A_\infty\ge A_{T_z^+}$ holds on the event $\{T_z^+<\infty\}$. So 
\be
0\le \ex_x[\exp(-A_\infty)\ind_{\{T_z^+<\infty\}}]\le \ex_x[\exp(-A_{T_z^+})\ind_{\{T_z^+<\infty\}}].\nn
\ee
If \eqref{eq:condition1} holds, then we have 
\be
\lim_{z\to\infty}\ex_x[\exp(-A_\infty)\ind_{\{T_z^+<\infty\}}]=0.\label{eq:42}
\ee
To prove Assumption \ref{ass:A} holds, we only need to demonstrate that, if $\pr(\limsup_{t\to\infty}X_t<\infty)=1$ fails to hold, then $\pr_x(A_\infty=\infty)=1$. But by \cite[Theorem VI.12]{Bertoin_1996}, we know that in this case $X_\cdot$ is either drifting to $\infty$, or oscillating, so $\pr_x(T_z^+<\infty)=1$ for all $z>x$. This implies that \eqref{eq:42} actually reads as $\ex_x[\exp(-A_\infty)]=0$, so $\pr_x(A_\infty=\infty)=1$.
\end{proof}

\begin{proof}[Proof of Lemma \ref{lem:limit}]
Consider the function
\[z\mapsto \ex[\et^{-A_{T_z^+}}f(X_{T_z^+})\ind_{\{T_z^+<\infty\}}]=f(z)\pr(\ol{X}_\zeta>z)=f(z)\exp(-\int_0^z\Lambda(y)\diff y),\quad\forall z>0.\] 
By the monotonicity of the function as discussed in Remark \ref{rmk:constantr}, and Assumption \ref{assumption'}(i), we know that the limit in \eqref{eq:defc0} exists, and the limit $c_0\in[0,\infty)$ if $x^\star<\infty$, and $c_0\in[0,\infty]$ if $x^\star=\infty$. To see why $c_0=0$ cannot happen with $x^\star=\infty$, notice that $f'(x)\ge \Lambda(x)f(x)$ for all $x\in\R$ in this case, and by Gr\"onwall's inequality, we have $f(z)\et^{-\int_0^z\Lambda(y)\diff y}\ge f(x)\et^{-\int_0^x\Lambda(y)\diff y}$ for all $z\ge x$. Choosing an $x$ sufficiently large such that $f(x)>0$ implies that the limit $c_0\ge f(x)\et^{-\int_0^x\Lambda(y)\diff y}>0$. 
\end{proof}

\begin{proof}[Proof of Lemma \ref{lem1}]
Using Corollary 2(i) and equation (11) of \cite{OccupationInterval}, we know that, for $0<\epsilon,x<c$,
\begin{align}
\ex_{x}[\et^{-q\int_{0}^{T_{c}^+}\ind_{(-\epsilon,\epsilon)}(X_{t})\diff t}\ind_{\{T_{c}^+<\infty\}}]
=\frac{\et^{\Phi(0)(x+\epsilon)}+q\int_{-\epsilon}^\epsilon W(x-y)\mathcal{H}^{(q)}(y+\epsilon)\diff y}{\et^{\Phi(0)(c+\epsilon)}+q\int_{-\epsilon}^\epsilon W(c-y)\mathcal{H}^{(q)}(y+\epsilon)\diff y},\label{eq:local}
\end{align}
where, for any $q\ge 0$, we defined
\be
\mathcal{H}^{(q)}(x)=\et^{\Phi(0)x}+q\int_{0}^{x}\et^{\Phi(0)(x-y)}W^{(q)}(y)\diff y.\label{eq:H}
\ee

To obtain the law of local time $L_{T_c^+} $, we use the occupation time density formula:
\[L_{T_c^+} =\lim_{\epsilon\downarrow 0}\frac{1}{2\epsilon}\int_0^{T_c^+}\ind_{(-\epsilon,\epsilon)}(X_t)\diff t, \quad \pr_x\text{-a.s. on the event }\{T_c^+<\infty\}.\]
In particular, by letting $q=\frac{r}{2\epsilon}$ and taking the limit as $\epsilon\downarrow 0$ in \eqref{eq:local}, we will get the result, thanks to the bounded convergence theorem. However, this limit requires a subtle estimate in order to properly control the $\mc{H}^{(q)}$ term in the above integrals. To that end, we recall the estimate of the $q$-scale functions appeared in the proof of \cite[Lemma 8.3]{Kyprianou2006}:
\be
0\le W^{(q)}(x)\leq\sum_{k\geq0}q^k\frac{x^k}{k!}W^{k+1}(x)=W(x)\et^{qxW(x)}.\nn
\ee
The above inequalities imply that,
for any $y\in(-\epsilon,\epsilon)$, by $0<\frac{y+\epsilon}{2\epsilon}<1$ we have $\frac{r}{2\epsilon}z<r$ for all $z\in(0,y+\epsilon)$, so
\begin{align*}
0\le \mc{H}^{(\frac{r}{2\epsilon})}(y+\epsilon)-\et^{\Phi(0)(y+\epsilon)}\le  \frac{r}{2\epsilon}\int_0^{y+\epsilon}\et^{\Phi(0)(y+\epsilon-z)}W(z)\et^{rW(z)}\diff z\le r\et^{2\Phi(0)\epsilon}W(2\epsilon)\et^{rW(2\epsilon)}.
\end{align*}
Hence,
\begin{align}
 \frac{r}{2\epsilon}\int_{-\epsilon}^\epsilon W(x-y)\mc{H}^{(\frac{r}{2\epsilon})}(y+\epsilon)\diff y\ge&\frac{r}{2\epsilon}\int_{-\epsilon}^\epsilon W(x-y)\et^{\Phi(0)(y+\epsilon)}\diff y\to rW(x),\text{ as }\epsilon\downarrow 0.\nn
\end{align}
On the other hand,
\begin{align*}
\frac{r}{2\epsilon}\int_{-\epsilon}^\epsilon W(x-y)\mc{H}^{(\frac{r}{\epsilon})}(y+\epsilon)\diff y\le &\frac{r}{2\epsilon}\int_{-\epsilon}^\epsilon W(x-y)\left(\et^{\Phi(0)(y+\epsilon)}+r\et^{2\Phi(0)\epsilon}W(2\epsilon)\et^{rW(2\epsilon)}\right)\diff y\nn\\
\to&rW(x)(1+rW(0)\et^{rW(0)})=rW(x), \text{ as }\epsilon\downarrow 0,
\end{align*}
thanks to the fact that $W(0)=0$.
The equation \eqref{eq:lem} now follows from taking the limit as $\epsilon\downarrow 0$ in both the numerator and the denominator of  \eqref{eq:local}.
\end{proof}

\begin{lem}\label{lem:inf}
Let $(a_n)_{n\ge1}$ and $(b_n)_{n\ge1}$ be two nonnegative sequences, and $(a_n)_{n\ge1}$ is bounded and nondecreasing. Then 
\[\liminf_{n\to\infty}(a_nb_n)=\lim_{n\to\infty}a_n\cdot\liminf_{n\to\infty}b_n.\]
\end{lem}
We omit the proof of Lemma \ref{lem:inf} as it is a standard exercise in real analysis. 

\begin{lem}\label{lem:Ulocal}
The process $(\exp(-A_t+\int_0^{X_t}\Lambda(y)\diff y))_{t\ge0}$ is a local martingale. 
\end{lem}
\begin{proof}
Let us denote $U_t=\exp(-A_t+\int_0^{X_t}\Lambda(y)\diff y)$.
Consider the sequence of stopping times $(T_n^+)_{n\ge1}$, which is almost surely increasing and diverging (i.e., $T_n^+\to\infty$, $\pr$-a.s.). Then for all $t\ge0$, we have 
\[\ex_x[U_{t\wedge T_n^+}]=U_0=\et^{\int_0^x\Lambda(y)\diff y}, \quad\forall n\le x;\]
for any fixed $x\in\R$ and all $n>x$, we have
\begin{align}
U_0&=\et^{\int_0^x\Lambda(y)\diff y}=\et^{\int_0^n\Lambda(y)\diff y}\ex_x[\et^{-A_{T_n^+}}\ind_{\{T_n^+<\infty\}}]=\ex_x[U_{T_n^+}\ind_{\{T_n^+<\infty\}}]\nn\\
&=\ex_x[\ex_x[U_{T_n^+}\ind_{\{T_n^+<\infty\}}|\mc{F}_{t\wedge T_n^+}]]=\ex_x[\et^{-A_{t\wedge T_n^+}}\cdot\et^{\int_0^n\Lambda(y)\diff y}\cdot\ex_{X_{t\wedge T_n^+}}[\et^{-A_{T_n^+}}\ind_{\{T_n^+<\infty\}}]]\nn\\
&=\ex_x[\exp(-A_{t\wedge T_n^+}+\int_0^{X_{t\wedge T_n^+}}\Lambda(y)\diff y)]=\ex_x[U_{t\wedge T_n^+}],
\end{align}
where the second line follows from the  tower property of conditional expectation and the strong Markov property of $X_\cdot$. Hence $(U_{t\wedge T_n^+})_{t\ge0}$ is a martingale for all $n\ge1$.
\end{proof}

\begin{lem}\label{lem:cont}
If $f(\cdot)$ is lower semi-continuous and belongs to $\Upsilon_{\e_r}$, then the value function $$V(x)=\sup_{\tau\in\mathcal{T}}\ex_x[\et^{-r\tau}f(X_\tau)\ind_{\{\tau<\infty\}}]$$ is nondecreasing and lower semi-continuous in $x$.
\end{lem}
\begin{proof}
Suppose there are a nondecreasing function $h(\cdot)$ such that $f(x)=\ex[h(x+\ol{X}_{\e_r})]$ and a constant $x^\star\in[-\infty,\infty)$ such that $h(x)>0$ if and only if $x>x^\star$.  Then by Theorem \ref{singlesolution}, we know that $V(x)=\ex[h(x+\ol{X}_{\e_r})\ind_{\{x+\ol{X}_{\e_r}>x^\star\}}]$, and $V(\cdot)\in\Upsilon_{\e_r}$, which, by Remark \ref{remark:32}, implies that $V(\cdot)$ is nondecreasing. 

Recall that, for nondecreasing functions, 
the lower semi-continuity is equivalent to the left continuity. In particular,  $f(\cdot)$ is left continuous. To show that $V(\cdot)$ is left continuous,  
consider any $x_1>x_2$, notice that the non-negative random variable defined as $\Delta(x_1,x_2)=h(x_1+\ol{X}_{\mathbf{e}_r})-h(x_2+\ol{X}_{\mathbf{e}_r})$ is decreasing in $x_2$ for each fixed $x_1$, and satisfies (by Fatou's lemma)
\begin{multline}0\le \ex[\lim_{x_2\uparrow x_1}\Delta(x_1,x_2)]\le \liminf_{x_2\uparrow x_1}\ex[h(x_1+\ol{X}_{\mathbf{e}_r})-h(x_2+\ol{X}_{\mathbf{e}_r})]=\liminf_{x_2\uparrow x_1}\left(f(x_1)-f(x_2)\right)\\=f(x_1)-\limsup_{x_2\uparrow x_1}f(x_2)
=f(x_1)-\lim_{x_2\uparrow x_1}f(x_2)=0,\end{multline}
where the last two steps are due to the left continuity of $f(\cdot)$. It follows that the nonnegative random variable $\lim_{x_2\uparrow x_1}\Delta(x_1,x_2)=0$, $\pr$-a.s. That is, $\lim_{x_2\uparrow x_1}h(x_2+\ol{X}_{\mathbf{e}_r})=h(x_1+\ol{X}_{\mathbf{e}_r})$, $\pr$-a.s. which implies that $\lim_{x_2\uparrow x_1}h(x_2+\ol{X}_{\mathbf{e}_r})\ind_{\{x_2+\ol{X}_{\mathbf{e}_r}>x^\star\}}=h(x_1+\ol{X}_{\mathbf{e}_r})\ind_{\{x_1+\ol{X}_{\mathbf{e}_r}>x^\star\}}$, $\pr$-a.s. On the other hand, because 
\begin{align}
0\le&h(x_1+\ol{X}_{\mathbf{e}_r})\ind_{\{x_1+\ol{X}_{\e_r}>x^\star\}}-h(x_2+\ol{X}_{\mathbf{e}_r})\ind_{\{x_2+\ol{X}_{\e_r}>x^\star\}}\le h(x_1+\ol{X}_{\mathbf{e}_r})\ind_{\{x_1+\ol{X}_{\e_r}>x^\star\}}\nn,
\end{align}
by the dominated convergence theorem, we know that the value function $V(x)=\ex[h(x+\ol{X}_{\mathbf{e}_r})\ind_{\{x+\ol{X}_{\e_r}>x^\star\}}]$ is left continuous at $x_1$, so it is also lower semi-continuous at $x_1$.   
\end{proof}

\section{Scale functions of spectrally negative L\'evy processes}\label{sec:SNLP}
Let $X_\cdot$ be a spectrally negative L\'evy process, with Laplace exponent  $\psi(\lambda)=\log \ex[\et^{\lambda X_{1}}]$, which is well defined for all real number $\lambda \ge0$ (see, e.g., \cite[page 78]{Kyprianou2006}).
For $q\geq0$, the $q$-scale function of the process $X_\cdot$ is the unique function supported on $[0,\infty)$, defined via the Laplace transform
\be
\int_{0}^{\infty}e^{-\lambda y}W^{(q)}(y)dy=\frac{1}{\psi(\lambda)-q},\text{ for }\lambda>\Phi(q),\label{eq:scale}
\ee
 and $\Phi:[0,\infty)\rightarrow[0,\infty)$ is defined by $\Phi(q)=\sup\{\lambda\geq0 : \psi(\lambda)=q\}$ such that
$\psi(\Phi(q))=q,\text{ }q\geq0.$
It is well-known that $W^{(q)}(\cdot)$ is strictly increasing on $[0,\infty)$, and is continuously differentiable over $(0,\infty)$ if the jump measure of $X_\cdot$ has no atoms.  Suppose $\psi'(\Phi(q))>0$, i.e., $q>0$ or $q=0$ and $\psi'(0)>0$, then as $x\to\infty$, we have
\be
\et^{-\Phi(q)x}W^{(q)}(x)\to1/\psi'(\Phi(q)).\label{limbigx}
\ee
Moreover, $W^{(q)}(0)=W^{(0)}(0)= W(0)\ge 0$,\footnote{When $q=0$ we suppress the superscript of the $0$-scale function.} where the last inequality becomes an equality if and only if $X_\cdot$ has paths of unbounded variation. In  case that $W^{(q)}$ is continuously differentiable on $(0,\infty)$, we have
\be
W^{(q)\prime}(0+)=\begin{dcases}\frac{2}{\sigma^2}, \quad &\text{if $\sigma>0$}\\
\infty, \quad&\text{if $\sigma=0$ and $\Pi(-\infty,0)=\infty$}\\\frac{q+\Pi(-\infty,0)}{\gamma^2},\quad&\text{else}\end{dcases}.\label{limsmallx}\ee
In addition, it is known that, (see, e.g., the proof of   \cite[Lemma 8.2]{Kyprianou2006}, and   \cite[Eq. (3.13)]{Leung_Yamazaki_2011}), the mapping $x\mapsto W^{(q)\prime}(x)/W^{(q)}(x)$ is strictly decreasing over $(0,\infty)$, with limit
\be
\lim_{x\to\infty}\frac{W^{(q)\prime}(x)}{W^{(q)}(x)}=\Phi(q).\label{eq:W'}
\ee

\bibliographystyle{amsplain}
\def\cprime{$'$}
\providecommand{\bysame}{\leavevmode\hbox to3em{\hrulefill}\thinspace}
\providecommand{\MR}{\relax\ifhmode\unskip\space\fi MR }
\providecommand{\MRhref}[2]{%
  \href{http://www.ams.org/mathscinet-getitem?mr=#1}{#2}
}
\providecommand{\href}[2]{#2}

\bibliographystyle{amsplain}

\begin{thebibliography}{10}

\bibitem{Alili2005}
L.~Alili and A.~E. Kyprianou, \emph{Some remarks on first passage of {L}\'evy
  processes, the {A}merican put and pasting principles}, The Annals of Applied
  Probability \textbf{15} (2005), no.~3, 2062--2080. \MR{MR2152253
  (2006b:60078)}

\bibitem{Bertoin_1996}
J.~Bertoin, \emph{{L}\'evy processes}, Cambridge Tracts in Mathematics, vol.
  121, Cambridge University Press, Cambridge, 1996. \MR{MR1406564 (98e:60117)}

\bibitem{Carmona2008}
R.~Carmona and N.~Touzi, \emph{Optimal multiple stopping and valuation of swing
  options}, Mathematical Finance \textbf{18} (2008), no.~2, 239--268.

\bibitem{CPT12}
M.~Ciss\'e, P.~Patie, and E.~Tanr\'e, \emph{Optimal stopping problems for some
  {M}arkov processes}, Annals of Applied Probability \textbf{22} (2012), no.~3,
  1243--1265.

\bibitem{Dayanik2008b}
S.~Dayanik, \emph{Optimal stopping of linear diffusions with random
  discounting}, Mathematics of Operations Research \textbf{33} (2008), no.~3,
  645--661. \MR{MR2442645 (2009i:60081)}

\bibitem{measure_integration81}
G.~Debarra, \emph{Measure theory and integration}, New Age International (P)
  Limited, 1981.

\bibitem{Leung_Yamazaki_2011}
E.~Egami, T.~Leung, and K.~Yamazaki, \emph{Default swap games driven by
  spectrally negative {L}\'{e}vy processes}, Stochastic Processes and their
  Applications \textbf{123} (2013), no.~2, 347--384. \MR{MR3003355}

\bibitem{CAF66}
R.K. Getoor, \emph{Continuous additive functionals of a {M}arkov proces with
  applications to processes with independent increments}, Journal of
  Mathematical Analysis and Applications \textbf{13} (1966), 132--153.

\bibitem{HLY14}
S-R Hsiau, Y-S Lin, and Y-C Yao, \emph{Logconcave reward functions and optimal
  stopping rules of threshold form}, Electronic Journal of Probability
  \textbf{19} (2014), no.~120, 1--18.

\bibitem{KaratzasShreve01}
I.~Karatzas and S.~Shreve, \emph{Methods of mathematical finance}, Springer,
  1998.

\bibitem{Kuznetsov_2011}
A.~Kuznetsov, Andreas~E. Kyprianou, and V.~Rivero, \emph{The theory of scale
  functions for spectrally negative {L}\'evy processes}, Springer Lecture Notes
  in Mathematics \textbf{2061} (2013), 97--186. \MR{MR3014147}

\bibitem{Kyprianou2006}
A.E. Kyprianou, \emph{Introductory lectures on fluctuations of {L}\'evy
  processes with applications}, Universitext, Springer-Verlag, Berlin, 2006.
  \MR{MR2250061 (2008a:60003)}

\bibitem{KyprSury05}
A.E. Kyprianou and B.A. Surya, \emph{On the {N}ovikov-{S}hiryaev optimal
  stopping problems in continuous time}, Electronic Communication in
  Probability \textbf{10} (2005), no.~15, 146--154.

\bibitem{LTZIJTAF15}
T.~Leung, K.~Yamazaki, and H.~Zhang, \emph{An analytic recursive method for
  optiomal multiple stopping: {C}anadization and phase-type fitting},
  International Journal of Theoretical and Applied Finance \textbf{18} (2015),
  no.~5, 1550032.

\bibitem{TimKazuHZ14}
\bysame, \emph{Optimal multiple stopping with negative discount rate and random
  refraction times under {L}\'evy models}, SIAM Journal on Control and
  Optimization \textbf{53} (2015), no.~4, 2373--2405.

\bibitem{Linetsky99}
V.~Linetsky, \emph{Step options}, Mathematical Finance \textbf{9} (1999),
  no.~1, 55--96.

\bibitem{OccupationInterval}
R.~Loeffen, J.-F. Renaud, and X.~Zhou, \emph{Occupation times of intervals
  until first passage times for spectrally negative {L}\'evy processes},
  Stochastic Processes and their Applications \textbf{124} (2014), no.~3,
  1408--1435. \MR{MR3148018}

\bibitem{mordecki2002}
E.~Mordecki, \emph{Optimal stopping and perpetual options for {L}\'{e}vy
  processes}, Finance and Stochastics \textbf{6} (2002), 473--493.
  \MR{MR1932381 (2003j:91059)}

\bibitem{NovkShir07}
A.~Novikov and A.N. Shiryaev, \emph{On a solution of the optimal stopping
  problem for processes with independent increments}, Stochastics: An
  International Journal of Probability and Stochastic Processes \textbf{79}
  (2007), no.~3-4, 393--406.

\bibitem{OmegaRZ15}
N.~Rodosthenous and H.~Zhang, \emph{Beating the {O}mega clock: an optimal
  stopping problem with random time-horizon under spectrally negative {L}\'evy
  models}, The Annals of Applied Probability (2018), forthcoming.

\bibitem{Sato99}
K.~Sato, \emph{{L}\'{e}vy processes and infinitely divisible distributions},
  Cambridge University Press, Cambridge, 1999.

\bibitem{Surya2007}
B.A. Surya, \emph{An approach for solving perpetual optimal stopping problems
  driven by {L}\'evy processes}, Stochastics: An International Journal of
  Probability and Stochastic Processes \textbf{79} (2007), no.~3-4, 337--361.
  \MR{MR2308080 (2008e:60114)}

\bibitem{YamazakiAMO2014}
K.~Yamazaki, \emph{Contraction options and optimal multiple-stopping in
  spectrally negative {L}\'{e}vy models}, Applied Mathematical and Optimization
  \textbf{72} (2014), no.~1, 147--185.

\bibitem{ZeghalSwing}
A.~B. Zeghal and M.~Mnif, \emph{Optimal multiple stopping and valuation of
  swing options in {L}\'evy models}, International Journal of Theoretical and
  Applied Finance \textbf{9} (2006), no.~8, 1267--1297.

\end{thebibliography}
\def\cprime{$'$}
\providecommand{\bysame}{\leavevmode\hbox to3em{\hrulefill}\thinspace}
\providecommand{\MR}{\relax\ifhmode\unskip\space\fi MR }
\providecommand{\MRhref}[2]{%
  \href{http://www.ams.org/mathscinet-getitem?mr=#1}{#2}
}
\providecommand{\href}[2]{#2}



\end{document}